\documentclass{article}
\usepackage[utf8]{inputenc}

\usepackage[affil-it]{authblk}

\usepackage{comment}
\usepackage{enumitem}
\usepackage{mathtools}

\usepackage{graphicx}
\usepackage{amsmath}
\usepackage{amsfonts}
\usepackage{amssymb}
\usepackage{amsthm}
\usepackage[colorlinks=true, allcolors=blue]{hyperref}

\usepackage {tikz}
\usetikzlibrary {positioning}
\tikzset{main node/.style={circle,fill=blue!20,draw,minimum size=2cm,inner sep=0pt},}
\tikzset{other node/.style={circle,fill=blue!20,draw,minimum size=1.5cm,inner sep=0pt},}
\tikzset{smaller dot/.style={fill=black,circle,scale=0.5}}

\newtheorem{prop}{Proposition}
\newtheorem{coro}{Corollary}
\newtheorem{defi}{Definition}
\newtheorem{lem}{Lemma}

\newcommand\twoheaduparrow{\mathrel{\rotatebox[origin=c]{90}{$\twoheadrightarrow$}}}

\newcommand\twoheaddownarrow{\mathrel{\rotatebox[origin=c]{270}{$\twoheadrightarrow$}}}

\newcommand\blfootnote[1]{%
  \begingroup
  \renewcommand\thefootnote{}\footnote{#1}%
  \addtocounter{footnote}{-1}%
  \endgroup
}

\title{Computation as uncertainty reduction:\\ a  simplified order-theoretic framework}

\author{Pedro Hack, Daniel A. Braun, Sebastian Gottwald}

\date{ }

\begin{document}

\maketitle

\begin{abstract}
    Although there is a somewhat standard formalization of computability on countable sets given by Turing machines, the same cannot be said about uncountable sets. Among the approaches to define computability in these sets, order-theoretic structures have proven to be useful. Here, we discuss the mathematical  structure needed to define computability using order-theoretic concepts. In particular, we introduce a more general framework and discuss its limitations compared to the previous one in domain theory. We expose four features in which the stronger requirements in the domain-theoretic structure allow to improve upon the more general framework: computable elements, computable functions, model dependence of computability and complexity theory.  Crucially, we show computability of elements in uncountable spaces can be defined in this new setup, and argue why this is not the case for computable functions. Moreover, we show the stronger setup diminishes the dependence of computability on the chosen order-theoretic structure and that, although a suitable complexity theory can be defined in the stronger framework and the more general one posesses a notion of computable elements, there appears to be no proper notion of element complexity in the latter. 
\end{abstract}

\section{Introduction}

The pioneering work of Turing \cite{turing1937computable,turing1938computable} has laid the foundations to a formal approach to computability on the natural numbers based on an idealized machine: the Turing machine \cite{rogers1987theory}. Several other attempts have been made to formalize computability on natural numbers, like lambda calculus \cite{barendregt1984lambda}, or universal register machines \cite{cutland1980computability}, but they all have been shown to be equivalent to Turing's approach. It is, however, not immediately clear how to translate computability from countable to uncountable spaces such as the real numbers. Several such attempts have been made \cite{weihrauch2012computable,weihrauch2012computability,abramsky1994domain}, but so far no canonical way of approaching computability on uncountable sets has been established. In consequence, for instance, there is no consensual formal definition of algorithm for the real numbers and, thus, it remains unclear what can be computed there. Actually, computability varies significantly from one model to another, in contrast with the situation for countable sets. For example, one can distinguish computability on the real numbers when it is introduced using their decimal (or any other base) representation with the one derived from the so-called \emph{negative-digit representation} \cite{di1996real,wiedmer1980computing,avizienis1961signed,boehm1987constructive}. The difference between them is significant since, in the former, which was introduced by Turing himself \cite{turing1937computable}, a simple operation like multiplication by $3$ is not computable, while it is computable in the latter \cite{di1996real}. Notice this matter is not only theoretical, but practical, as a wide variety of uncountable structures where computability notions are needed are constantly introduced, for example, in physics \cite{pour2017computability,edalat1999computable,edalat1997domains}.



Among the most extended approaches to computability on uncountable spaces are \emph{type-2 theory of effectivity (TTE)}, where a variation of the Truing machine is used to formalize computability \cite{weihrauch2012computability,kreitz1985theory}, and \emph{domain theory}, where order-theoretical notions are used to extend computability from Turing machines to uncountable sets \cite{scott1970outline,ershov1972computable,abramsky1994domain,keimel2017domain}, in particular, to the real numbers \cite{edalat1999domain,di1996real}.

The introduction of the order structure in domain theory is usually done relying on a somewhat intuitive picture consisting of the two notions of \emph{convergence} and \emph{approximation} that are captured by two order structures, $\preceq$ and $\ll$, respectively \cite{scott1970outline,abramsky1994domain,martin2000foundation}. Notice, both of these orderings can be considered as two different degrees of \emph{information} of one element in relation to another. Even though this results in a rich computability framework, dealing with both of these structures at the same time hides the separate effects of each order structure. Here, we explore a simplified and more general approach by only considering the convergence ordering $\preceq$, which we refer to as \emph{uncertainty reduction} or \textit{information}. The resulting computability framework trades order structure for a more limited notion of computability, which at the same time makes the fundamental underpinnings of domain theory more transparent.


More specifically, we begin, in Section \ref{newton}, developing an intuition through a simple algorithm: the bisection method. Right after, in Section \ref{order in compu}, we introduce formally the more general framework. In Section \ref{uniform compu}, we recall the order-theoretical approach in domain theory. We finish, in Section \ref{comparing apporaches}, differentiating the properties of computability in both approaches through four criteria: computable elements, computable functions, model dependence of computability and complexity.

\section{The intuitive picture: the bisection method}
\label{newton}

Before formally defining a general order-theoretic framework that allows the introduction of computability on uncountable spaces, we give some intuition through the example of the bisection method. We formally develop, in the following section, the computational framework presented here.

Assume we intend to use the bisection method to calculate the single zero $x$ of, say, a polynomial with rational coefficients $p: \mathbb{R} \to \mathbb{R}$. The algorithm begins with two rational points where $p$ takes a different sign and evaluates the polynomial at their midpoint $m$. If $p(m)=0$, then the algorithm stops and outputs $m$. Otherwise, it takes $m$ and one of the previous rational points, the one where $p$ has a different sign than $p(m)$, and repeats the process. Now, we could say, $x$ is computable  if there exists such a $p$ whose single zero is $x$, since we can apply the bisection method to $p$ and obtain an arbitrarily precise description of $x$.
Note that two features are fundamental for this assertion: first, we can write the instructions for the bisection method using finite space with all the included operations being doable in finite time, and, second, the bisection method is a process of uncertainty reduction regarding $x$,
that is, we begin knowing that $x$ belongs to a certain interval with rational endpoints and, at every step of the process, we reduce the length of the interval by half (or, in the best scenario, reduce the uncertainty completely and obtain $x$). Importantly, the degree of uncertainty after every step can be compared to that of another step by simply comparing the length of the intervals where the following step would be performed. 

However, instead of simply considering the zeros of rational polynomials, we would like to have a more general theory of computation on uncountable spaces. Just like in the bisection method, this requires a notion of finite instructions and of uncertainty reduction along the computation. We show formally how this can be done in the following section, where we introduce a new order-theoretic framework to this end. 


Let us, however, first reexpress the example regarding the bisection method in a way that can then be easily extended to the general case. The outcome of each computation step in the bisection method results in an interval $I_i = [a_i,b_i]$ with rational endpoints $a_i,b_i\in\mathbb Q$. Thus, here, a computation results in a sequence inside the set $Y$ of all compact intervals on the real line. Actually, by construction, the bisection method will always result in a sequence $(I_i)_i$ inside the smaller subset $B\subset Y$ of intervals with rational endpoints. There are two important properties of these sequences, which we will describe in the following.

The sequences $(I_i)_i\subset B$ corresponding to an application of the bisection method have a \emph{finite description}, both in the sense that each element $I_i$ of the sequence is given by its rational endpoints, which themselves have a finite description, and, moreover, in the sense that the sequence $(I_i)_i$ of such intervals is generated by using finite instructions (the bisection method). The crucial property for generalizing this to arbitrary spaces in the next section is that there exists an \textit{effectively calculable} (c.f. Definition \ref{def:finite map} and the discussion afterwards) one-to-one map $\alpha:\mathbb N\to B$, which, in the case when $B$ consists of intervals with rational endpoints, can be easily constructed from the maps that classically have been used to show that the rational numbers are countable. In general, such a map will not only provide the finite description of the elements of $B$ but also translates the finite description of sequences in $\mathbb N$ (given by Turing machines) to sequences in $B$.  

The other crucial property of the sequences $(I_i)_i$ that allows to decide about computability of a number $x\in\mathbb R$ in the bisection method is that of \textit{convergence}, \emph{information}, or \textit{uncertainty reduction}, which is the property that $I_{i+1}\subseteq I_i$ for all $i$. In particular, any $x\in X \coloneqq \mathbb R$ can be identified with the singleton $[x,x]\in Y$, which can be produced by the bisection method if it is the infimum of some sequence $(I_i)_i$ in $B$, $[x,x] = \inf_i I_i$, in the sense of set inclusion. This means, we are relating the sequences $(I_i)_i$ in $B$ with the set of singletons, and thus with $X$, by means of the partial map $f:Y\to X, [x,x] \mapsto x$ (\textit{partial} means that it is not necessarily defined on all of $Y$). The question of whether a certain element $y \in Y$ can be computed or not, and, hence, whether some $x \in X$ can be computed, is therefore not only related to how rich $B$ is, in the sense that whether $B$ contains a sequence with a finite description that has $y$ as infimum, but also to how rich the ordering is itself in order to produce non-trivial infima.

Now, in order to make the jump to the general case, we essentially need to replace:
\begin{itemize}
\item $X$, $Y$ by general sets where we define computabilty (with requirements to be specified in the next section),
\item the partial map $f$ from $Y$ to $\mathbb R$, here $f\big([x,x]\big)= x$, by a general surjective partial map $\rho:Y\to X$,
\item $B$ by a more general countable subset of $Y$, a so-called (weak) basis, containing sequences $(b_i)_i$ that have a finite discription and to which we can associate elements $y \in Y$ such that $\rho(y)$ is defined.
\item the finite instructions used to generate sequences in $B$ (the bisection method) by Turing machines in order to specify which sequences in $\mathbb N$ can be finitely described, and translate them to $B$ using an effectively calculable one-to-one map $\alpha:\mathbb N\to B$,
\item (inverted) set inclusion $\supseteq$ by a general partial order $\preceq$ on $Y$ that specifies the reduction of uncertainty, and thus determines which points in $Y$ are suprema of sequences in $B$.
\end{itemize}
Actually, in the general case, we are not requiring uncertainty reduction in every computational step, but only that it is eventually reduced, i.e. for all $i,j$ there exists some $k$ such that $b_i\preceq b_k$ and $b_j\preceq b_k$, for a sequence $(b_i)_i\subset B$. Moreover, unlike $f: Y \to \mathbb R$, we allow for the elements of $X$ to be associated to more than one element in $Y$ through $\rho$.

In summary, we say that $x\in X$ is computable if there exists some $y \in Y$ such that $x=\rho(y)$, with $\rho:Y\to X$ being a surjective partial map, and $y$ is the supremum of some $(b_i)_i\subseteq B$ that eventually reduces uncertainty in the sense of a partial order $\preceq$ and has a finite discription, i.e., that $\alpha^{-1}((b_i)_i)$ is the output of some Turing machine.














\section{Computability via order theory}
\label{order in compu}

Let $X$ be a set where we intend to introduce computability. In order to do so, we will define a structure which carries computability from Turing machines, \emph{directed complete partial orders} with an \emph{effective weak basis}, and then translate computability from a representative $P$ of our structure to $X$ via a surjective partial map $\rho: P \rightarrow X$. We will first define the structure with its computability properties and only return to $\rho$ at the end of this section.

Before introducing directed complete partial orders, we
include some definitions
from the formal approach to computability on $\mathbb{N}$
based on Turing machines \cite{rogers1987theory,turing1937computable}.
A function $f:\mathbb{N} \rightarrow \mathbb{N}$ is \emph{computable} if there exists a Turing machine which, $\forall n \in \mathbb{N}$, halts on input $n$, that is, finishes after a finite amount of time, and returns $f(n)$. Note what we call a computable function is also referred to as a \emph{total recursive function} to differentiate it from functions where $\text{dom} (f) \subset \mathbb{N}$ holds \cite{rogers1987theory}, which we call \emph{partially computable}.
A subset $A \subseteq \mathbb{N}$ is said to be \emph{recursively enumerable} if either $A=\emptyset$ or there exists a computable function $f$ such that $A=f(\mathbb{N})$. Recursively enumerable sets are, thus, the subsets of $\mathbb{N}$ whose elements can be produced in finite time, as we can introduce the natural numbers one by one in increasing order in a Turing machine and it will output one by one, each in finite time, all the elements in $A$ (possibly with repetitions). Note there exist subsets of $\mathbb{N}$ which are not recursively enumerable \cite{rogers1987theory}. As we are also interested in computability on the subsets of $\mathbb{N}^2$, we translate the notion of recursively enumerable sets from $\mathbb{N}$ to $\mathbb{N}^2$ using pairing functions. A \emph{pairing function} $\langle \cdot,\cdot\rangle $ is a computable bijective function $\langle \cdot,\cdot\rangle: \mathbb{N} \times \mathbb{N} \rightarrow \mathbb{N}$.\footnote{Note we can define computable functions $f:\mathbb{N} \times \mathbb{N} \rightarrow \mathbb{N}$ via Turing machines, analogously to how we defined computable functions $f:\mathbb{N} \rightarrow \mathbb{N}$.} Since it is a common practice \cite{rogers1987theory}, we fix in the following $\langle n,m \rangle= \frac{1}{2} ( n^2+ 2nm + m^2 + 3n + m)$, the \emph{Cantor pairing function}. Moreover, we denote by $\pi_1$ and $\pi_2$ the computable functions that act as inverses of the Cantor pairing function, that is, $\pi_1(\langle n,m \rangle)=n$ and   $\pi_2(\langle n,m \rangle)=m$. Before continuing, we introduce an important concept for the following, finite maps. 

\begin{figure}[!tb]
\centering
\begin{tikzpicture}
    \node[other node] (1) {$\perp$};
    \node[other node] (2) [above right = 1cm and 1.5cm  of 1]  {$1$};
    \node[other node] (3) [above left = 1cm and 1.5cm  of 1]  {$0$};
    \node[other node] (4) [above left = 1cm  of 3]  {$00$};
    \node[other node] (5) [above right = 1cm  of 3]  {$01$};
    \node[other node] (6) [above left = 1cm  of 2]  {$10$};
    \node[other node] (7) [above right = 1cm  of 2]  {$11$};

   \path[draw,thick,->]
    (1) edge node {} (2)
    (1) edge node {} (3)
    (3) edge node {} (4)
    (3) edge node {} (5)
    (2) edge node {} (6)
    (2) edge node {} (7)
    ;
    \path[draw,thick,dotted]
    (4) edge node {} (-6,7)
    (4) edge node {} (-3,7)
    (5) edge node {} (-3,7)
    (5) edge node {} (0,7)
    (6) edge node {} (0,7)
    (6) edge node {} (3,7)
    (7) edge node {} (3,7)
    (7) edge node {} (6,7)
    ;
    \draw [line width=0.75mm] (-6.0,7.0) -- (6.0,7.0);
    \node [font=\fontsize{15}{15}] at (0,7.5) {$\{0,1\}^*$};
\end{tikzpicture}
\caption{Representation of the Cantor domain with alphabet $\Sigma=\{0,1\}$. Note an arrow from an element $w$ to an element $t$ represents $w \preceq_C t$. The crossing points of the dotted lines starting at an element and the horizontal line delimit the subset of $\{0,1\}^*$ that is above the element regarding $\preceq_C$.}
\label{cantor dom}
\end{figure}
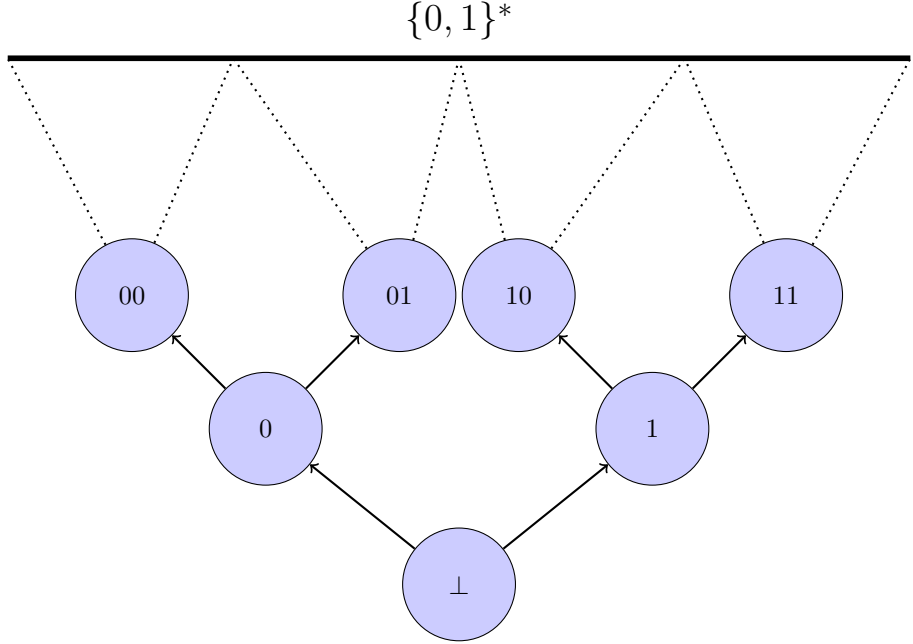

\begin{defi}[Finite map]
\label{def:finite map}
We say a map $\alpha: \text{dom}(\alpha) \to A$, where $\text{dom}(\alpha) \subseteq \mathbb{N}$, is a finite map for $A$ or simply a finite map if $\alpha$ is bijective and both $\alpha$ and $\alpha^{-1}$ are effectively calculable.
\end{defi}

Notice the definition of finite maps relies on the informal notion of \emph{effective calculability}. This is mandatory as no general formal definition for \emph{computable} maps $\alpha: \mathbb{N} \to A$ exists \cite{cutland1980computability}. In fact, the struggle between formal and informal notions of computability, best exemplified by Church's thesis \cite{rogers1987theory}, lies at the core of computability theory. Finite maps
are also known as \emph{effective denumerations} \cite{cutland1980computability} or \emph{effective enumerations} \cite{scott1970outline}.
Moreover, finite maps are related to recursively enumerable sets since, given two sets $A,B \subseteq \mathbb{N}$ and a finite map $\alpha:A \to B$, $B$ ($A$) is recursively enumerable whenever $A$ ($B$) is, since, by Curch's thesis \cite{rogers1987theory}, we can simply add a finite description of $\alpha$ ($\alpha^{-1}$) to the Turing machine whose output is $A$ ($B$). Note that this is no longer true if $\alpha$ is not a finite map.

We define now the order structure which we rely on to introduce computability in $X$ via $\rho$ and connect it, right after, with Turing machines. A \emph{partial order} $\preceq$ on a set $P$ is a reflexive ($x \preceq x$ $\forall x \in P$), transitive ($x \preceq y$ and $y \preceq z$ imply $x \preceq z$ $\forall x,y,z \in P$) and antisymmetric ($x \preceq y$ and $y \preceq x$ imply $x = y$ $\forall x,y \in P$) binary relation. We will call a pair $(P,\preceq)$ a \emph{partial order} and denote it simply by $P$. We may think of $P$ as a set of data and of $\preceq$ as a representation of the precision relation between different elements in the set. Given $x,y \in P$, we may read $x \preceq y$ as \emph{$y$ is at least as informative as $x$} or as \emph{$y$ is at least as precise as $x$}. We intend now to introduce the idea of some $x \in P$ being the limit of other elements in $P$, that is, the idea that one can generate some element in $P$ via a process which outputs other elements of $P$. This notion is formalized by the least upper bounds of directed sets \cite{abramsky1994domain}. $A \subseteq P$ is a \emph{directed set} if, given $a,b \in A$, there exist some $c \in A$ such that $a\preceq c$ and $b\preceq c$. If $A \subseteq P$ is a directed set, then $b \in P$ is the \emph{least upper bound of $A$} if $a \preceq b$ $\forall a \in A$ and, given any $c\in P$ such that $a \preceq c$ $\forall a \in A$, then $b \preceq c$ holds. We denote the least upper bound of $A$ by $\sqcup A$ and also refer to it as the \emph{supremum of $A$}. Hence, we can generate some $x \in P$ by generating a directed set $A$ whose upper bound is $x$, $x = \sqcup A$. We have restricted ourselves to directed sets since we can think of them as the output of some computational process augmenting the precision or information given that, for any pair of outputs, there is a third which contains their information and, potentially, more. Directed sets are, thus, a formalization of a computational process having a direction, that is, processes gathering information in a consistent way.
Of particular importance are \emph{increasing sequences} or \emph{increasing chains}, subsets $A \subseteq P$ where $A=(a_n)_{n\geq0}$ and $a_n \preceq a_{n+1}$ $\forall n \geq 0$. We can interpret increasing sequences as the output of some process where information increases every step.

Any process whose outputs increase information should tend towards some element in $P$, that is, any directed set $A \subseteq P$ should have a supremum $\sqcup A \in P$. A partial order with such a property is called \emph{directed complete} or a \emph{dcpo}. Note that there are partial orders which are not directed complete, such as  $P_0\coloneqq((0,1),\leq)$, since there are computational processes in $P_0$ whose information increases and are headed nowhere in $P_0$, for example, one whose output is $\mathbb{Q} \cap [0,1)$.

Some subsets $B \subseteq P$ are able to generate all the elements in $P$ via the supremum of directed sets contained in $B$. We refer to them as \emph{weak bases}.

\begin{defi}[Weak basis]
\label{def:weak basis}
A subset $B \subseteq P$ of a dcpo $P$ is a weak basis if, for each $x \in P$, there exists a directed set $B_x \subseteq B$ such that $x=\sqcup B_x$.
\end{defi}

We are particularly interested in dcpos where countable weak bases exist, since we intend to inherit computability from Turing machines. In case we have some computational process whose outputs are in $B$ and which is approaching some $x \in P/B$, we would like to be able to provide, after a finite amount of time, the best approximation of $x$ so far. In order to do so, we need to distinguish the outputs we already have in terms of precision. This is possible if the weak basis is \emph{effective}.

\begin{defi}[Effective weak basis]
\label{def:eff weak basis}
A countable weak basis $B \subseteq P$ of a dcpo $P$ is effective if there exist both a finite map for $B=(b_n)_{n\geq0}$ and a computable function $f:\mathbb{N} \to \mathbb{N}$ such that $f(\mathbb{N}) \subseteq \{\langle n,m \rangle| b_n \preceq b_m\}$ and, for each $x \in P$, there is a directed set $B_x \subseteq B$ such that $\sqcup B_x = x$ and, if $b_n,b_m \in B_x/\{x\}$,
then there exists some $b_p \in B_x$ such that $b_n,b_m \prec b_p$ and $\langle n,p \rangle, \langle m,p \rangle \in f(\mathbb{N})$.
\end{defi}

Note we may show a countable weak basis $B=(b_n)_{n\geq0}$ is effective by proving the stronger property that the set
\begin{equation*}
\label{eff weak basis}
\{\langle n,m \rangle| b_n \preceq b_m\}
\end{equation*}
is recursively enumerable (see Proposition \ref{example} for an example). If this stronger condition is satisfied, then, for any finite subset $(b_n)_{n=1}^N \subseteq B$ where all elements are related, we can find some $n_0 \leq N$ such that $b_n \preceq b_{n_0}$ $\forall n \leq N$ and, since $\alpha: \mathbb{N} \to B$ is a finite map, we can determine the best approximation so far, $\alpha(n_0)$. The reason why we take the weaker approach in Definition \ref{def:eff weak basis} will
become clear after Proposition \ref{weak basis compu and no basis compu}. We define now computable elements for dcpos with an effective weak basis.

\begin{defi}[Computable element]
\label{def:compu ele}
If $P$ is a dcpo, $B\subseteq P$ is an effective weak basis and $\alpha$ is a finite map for $B$, then
an element $x \in P$ is computable if there exists some $B_x \subseteq B$ such that the properties in Definition \ref{def:eff weak basis} are fulfilled and $\alpha^{-1}(B_x) \subseteq \mathbb{N}$ is recursively enumerable.
\end{defi}

Notice, Definition \ref{def:compu ele} generalizes the way in which Turing introduced computability on $P_{inf}(\mathbb{N})$, the set of infinite subsets of $\mathbb{N}$ \cite{turing1937computable}. Take $P=\mathbb{N} \cup P_{inf}(\mathbb{N})$ where $n \preceq m$ if $n \leq m$ $\forall n,m \in \mathbb{N}$, $n\preceq A$ if $n \in A$ $\forall n \in \mathbb{N}, A \in P_{inf}(\mathbb{N})$ and $A \preceq A$ $\forall A \in P_{inf}(\mathbb{N})$. Using the identity as finite map, we have $B=\mathbb{N}$ is an effective weak basis for $P$ and $A \in P_{inf}(\mathbb{N})$ is computable if and only if $B_A = \{n \in \mathbb{N}|n \in A\}$ is recursively enumerable, which is precisely the formal definition using Turing machines. Notice, also, given an effective weak basis, the computable elements are independent of the chosen finite map (see Proposition \ref{indep finite map} in Section \ref{model dep comp}).
However, the same does not seem to be true for effective weak bases and no general hypothesis under which this property holds seems to exist (see Section \ref{model dep comp} for an informal discussion). We may say, thus, that $x \in P$ is $B$-computable, instead of just computable, where $B \subseteq P$ is a weak basis.
 


To recapitulate, the main features of our picture are $(a)$ a map from the natural numbers to some countable set of \emph{finite} labels $B=(b_n)_{n\geq0}$ and $(b)$ a partial order $\preceq$ which can be somewhat encoded via a Turing machine and which allows us to both associate to some infinite element of interest $x \in P$ a subset of our labels $B_x \subseteq B$ which converges to it and, in some sense, to decide what the best approximation of $x$ in $B_x$ is. As a result, the computability of $x$ reduces to whether some $B' \subseteq B$ such that $\sqcup B'=x$ can be finitely described.

Note that we could ask for a weaker structure, where, instead of a weak basis, we require the existence of some countable subset $B \subseteq P$ such that for each $x \in P$ there exists some $B_x \subseteq B$ where $\sqcup B_x=x$. This could, however, lead to pathological cases where it is impossible, after a finite amount of time, to gather the information acquired during the process in a consistent and finite way, that is, to give an approximate solution which profits from all the computational resources spent. Take, for example, $P=([0,2],\preceq)$, where $x \preceq y$ if either $y=2$ or $x \leq y$ with $x,y \in [0,1]$ or $x,y \in (1,2)$ with $x \leq y$, and take $B_x=(\mathbb{Q} \cap P)/\{1,2\}$. It is clear $\sqcup B_x=2$. However, since there is no $z \in B$ such that $x,y \preceq z$ whenever $x \in [0,1)$ and $y \in (1,2)$, then, after any finite amount of time, there is no finite representation joining the information in both $x$ and $y$, which
would not occur, by definition, if $B_x$ was directed.

The structure of the partial order $P$ is fundamental to address higher type computability. In case $P$ is \emph{trivial}, that is, $x \preceq y$ if and only if $x=y$ $\forall x,y \in X$ \cite{abramsky1994domain}, we cannot extend computability beyond countable sets and we end up considering countable sets with finite maps towards the natural numbers. This situation, thus, reduces our approach to the \emph{theory of numberings} \cite{ershov1999theory,badaev2000theory,badaev2008computability}. There, where the definition of finite maps is less restrictive, the question regarding the most suitable enumeration of a given countable set $A$ is addressed. In order to do so, a preorder, the \emph{reducibility relation}, is introduced on the set of possible finite maps on $A$  \cite{ershov1999theory}. However, in our narrower definition, all finite maps are equivalent in the reducibility relation, eliminating, thus, the ambiguity in their choice. Moreover, in case we take the trivial partial order on the natural numbers and use the identity as $\alpha$, we obtain a trivial computational model which simply tells us every natural number is computable. Note, while the set of computable elements in a dcpo is countable, since the set of recursively enumerable subsets of $\mathbb{N}$ is countable \cite{rogers1987theory}, the cardinality of a dcpo with an effective weak basis is bounded by the cardinality of the continuum $\mathfrak{c}$, as we show in Proposition \ref{cardinality where computability} in the Appendix \ref{proofs}. In fact, it is in the uncountable case where the order structure is of interest, since the theory of numberings is insufficient.

To conclude, a dcpo with an effective weak basis $P$ and a partial surjective map $\rho: P \to X$ can be used introduce computability on a set $X$. We say $x \in X$ is \emph{$\rho$-computable} if there exists some computable element $p \in P$ such that $x=\rho(p)$.
As computability on $X$ relies on the map $\rho$, it is fundamental to understand and classify the influence of using a specific $\rho$. While we do not address this issue here, we include, in Section \ref{uniform compu}, some references which deal with it in domain theory, a stronger order-theoretic approach to computability. Note we can introduce computability in $X$ only if $|X|\leq \mathfrak{c}$, as any dcpo $P$ with a countable weak basis fulfills $|P| \leq \mathfrak{c}$ (see Proposition \ref{cardinality where computability} in the Appendix \ref{proofs}) and $\rho$ is surjective.

\subsection{Example: the Cantor domain}
\label{examples}

To illustrate the abstract approach in this section, we complement Section \ref{newton} with another example of a dcpo with effective weak bases that will be relevant in the following.
Note that more examples with relevant applications can be found in \cite[Section 2.1]{hack2022relation} (see also the references therein).

If $\Sigma$ any finite set of symbols, an \emph{alphabet}, we denote by $\Sigma^*$ the set of finite strings of symbols in $\Sigma$ and by $\Sigma^\omega$ the set of countably infinite sequences of symbols. The union of these last two sets
is called the \emph{Cantor domain} or the \emph{Cantor set model} \cite{blanck2008reducibility,martin2000foundation} when we equip it with the prefix order, that is, the Cantor domain is the pair $(\Sigma^{\infty},\preceq_C)$ 
\begin{equation}
\label{Cantor domain}
\begin{split}
    \Sigma^\infty &\coloneqq \Big\{x\Big|x:\{1,..,n\} \to \Sigma,\text{ }0\leq n\leq \infty\Big\} \\
    x \preceq_C y &\iff |x| \leq |y| \text{ and } x(i)=y(i) \text{ } \forall i \leq |x|,
   \end{split}
\end{equation}
where $|s|$ is the cardinality of the domain of $s \in \Sigma^\infty$ and $|\Sigma|<\infty$.
We show $\Sigma^*$ is an effective weak basis for $\Sigma^{\infty}$ in Proposition \ref{example} in the Appendix \ref{proofs}. Note, in the simplest case where $\Sigma \coloneqq \{0,1\}$, the Cantor domain reduces to the binary representation of real numbers in the interval $[0,1]$ (see Figure \ref{cantor dom} for a representation of the Cantor domain in this case). In the same scenario, if we consider $\perp$ to be a proper initial interval $[a,b]$ with $a,b \in \mathbb{Q}$ for a family of polynomials with rational coefficients $(p_i)_{i \in I}$ that possess a unique zero, then the real numbers that can be computed using the bisection method on that family of polynomials are also computable in the Cantor domain. We can formalize this through $\rho: \Sigma^\infty \to X$, where $X$ consists of both the real numbers between $a$ and $b$ plus the intervals one may obtain applying the bisection method on $(p_i)_{i \in I}$ starting with $[a,b]$, and $\rho$ sends the finite strings to the intervals one may obtain starting the bisection method at $[a,b]$ and the infinite strings to the real numbers between $a$ and $b$ (for example, associating the zeros in the strings as representing that the bisection method transitions from an interval to its left half and the ones as transitions to its right half). Alternatively, we could take the $X$ as in Section \ref{newton} and use the restriction of $\rho$ to that subset as a map.

\section{Uniform computability via order theory}
\label{uniform compu}
In this section, we recall the more specific order-theoretic approach to computation in domain theory \cite{abramsky1994domain,edalat1999domain,stoltenberg1994mathematical,stoltenberg2001notes,martin2000foundation,scott1982lectures,mislove1998topology,cartwright2016domain}, which we refer to as the \emph{uniform computability} in order to discriminate it from our \emph{non-uniform} approach, where we define computability of an element of the existance of \emph{any} computational path, whereas here, there has to be a \emph{unique} one. 

After doing so, we highlight the difference to the more general approach we introduced before, which will be the main topic of discussion in the following sections.

In Section \ref{order in compu}, we considered an element $x$ in some dcpo $P$ with a countable weak basis $B\subseteq P$ to be computable if there existed some directed set $B_x \subseteq B$ such that $\sqcup B_x =x$ and whose associated subset of the natural numbers $\alpha^{-1}(B_x)$ was recursively enumerable.
We take now a stronger approach where we
associate to each $x \in P$ a unique directed set $B_x \subseteq B$ such that the fact $\alpha^{-1}(B_x)$ is recursively enumerable
is equivalent to $x$ being computable. In order to do so, $B_x$ should be fundamentally related to $x$, that is, the information in every $b \in B_x$ should be gathered by any process which computes $x$. Importantly, while the more general approach only enabled to introduce computable elements, computable functions can also be defined in this stronger framework, which was introduced by Scott in \cite{scott1970outline}.

 Before we continue, 
 we introduce the Scott topology, which will play a major role in the following. If $P$ is a dcpo, we say a set $O \subseteq P$ is \emph{open} in the \emph{Scott topology}
 if it is \emph{upper closed} (if $x \in O$ and there exists some $y \in P$ such that $x \preceq y$, then $y \in O$) and \emph{inaccessible by directed suprema} (if $A\subseteq P$ is a directed set such that $\sqcup A \in O$, then $A \cap O \neq \emptyset$) \cite{abramsky1994domain,scott1970outline}. We denote by $\sigma(P)$ the Scott topology on $P$.
 Note
 the Scott topology characterizes the partial order in $P$, that is,
\begin{equation}
\label{charac order by topo}
    x \preceq y \iff  x \in O \text{ implies } y \in O \text{ } \forall O \in \sigma(P)
\end{equation}
\cite[Proposition 2.3.2]{abramsky1994domain}.
Note, also, the Scott topology in $T_0$ since, if $x,y \in P$ and $x \neq y$, then, by antisymmetry, either $\neg(x\preceq y)$ or $\neg(y \preceq x)$ holds. Thus, by \eqref{charac order by topo}, there exists some $O \in \sigma(P)$ such that $x \in O$ and $y \not \in O$ or vice versa. We can think of the Scott topology as the family of properties on the data set $P$ \cite{smyth1983power}. In particular, by definition, if some computational processes has a limit, then any property of the limit is verified in finite time and, since $\sigma(P)$ is $T_0$, these properties are enough to distinguish between the elements of $P$. This feature is particularly important for the introduction of computability in uncountable sets whenever the Scott topology is second countable, since it implies the elements in a (potentially) uncountable data set $P$ can be distinguished through a countable set of properties.

In order to attach
to each element in a dcpo a unique subset of $\mathbb{N}$ which is equivalent to it for all computability purposes, we introduce the way-below relation. If $x,y \in P$, we say $x$ is \emph{way-below} $y$ or $y$ is \emph{way-above} $x$ and denote it by $x \ll y$ if, whenever $y \preceq \sqcup A$ for a directed set $A\subseteq P$, then there exists some $a \in A$ such that $x \preceq a$ \cite{abramsky1994domain,scott1972continuous}. The set of element way-below (way-above) some $x$ is denoted by $ \twoheaddownarrow x$ ($\twoheaduparrow x$). Note $x \ll y$ implies $x \preceq y$ and, if $x \preceq y$ and $y \ll z$, then $x \ll z$ $\forall x,y,z\in P$. We think of an element way-below another as containing \emph{essential information} about the latter, since any computational process producing the latter cannot avoid gathering the information in the former. For example, if $A_1 \subseteq \twoheaddownarrow x$ is a directed set such that $\sqcup A_1 =x$, then, if $A_2 \subseteq P$ is a directed set where $\sqcup A_2=x$, there exists $\forall a_1 \in A_1$ some $a_2 \in A_2$ such that $a_1 \preceq a_2$.
We can regard $A_1$, thus, as a canonical computational process that yields $x$ (this will be clarified in Proposition \ref{compu sets}). Note we are not requiring each element in $A_1$ to be in $A_2$, only their information. Actually, we can even have disjoint directed sets $A_1 \cap A_2 = \emptyset$ where $A_1,A_2 \subseteq \twoheaddownarrow x$ and $\sqcup A_1=\sqcup A_2 =x$. Take, for example, the dcpo $P \coloneqq ((0,1],\leq)$ and $A_1 \coloneqq \{q \in (0,1] \cap \mathbb{Q}|q<x\}$, $A_2 \coloneqq \{r \in (0,1]/\mathbb{Q}|r<x\}$ for any $x \in P$.

We introduce now, in a canonical way, computable elements in a dcpo. To do so, we first define effective bases. A subset $B \subseteq P$ is called a \emph{basis} if, for any $x \in P$, there exists a directed set $B_x \subseteq \twoheaddownarrow x \cap B$ such that $\sqcup B_x = x$ \cite{abramsky1994domain}. Note, if $B$ is a basis, then it is a weak basis and $(\twoheaduparrow b)_{b \in B}$ is a topological basis for $\sigma(P)$ \cite{abramsky1994domain}. A dcpo is called \emph{continuous} or a \emph{domain} if bases exist. As in the case of weak bases, we are interested in dcpos with a countable basis or \emph{$\omega$-continuous} dcpos, since computability can be introduced via
Turing machines there. Note $\Sigma^\infty$ is $\omega$-continuous, as $B=\Sigma^*$ is a countable basis.
We say a basis $B \subseteq P$ is \emph{effective} if there is a finite map which enumerates $B$, $B=(b_n)_{n\geq0}$, there is a bottom element $\perp \in B$, i.e. $\perp \preceq x$ for all $x \in P$, and
\begin{equation*}
\{\langle n,m \rangle| b_n \ll b_m\}
\end{equation*}
is recursively enumerable \cite{edalat1999domain,stoltenberg2008computability}. We assume w.l.o.g. $b_0 = \perp$ in the following. Note $\Sigma^*$ is, in fact, an effective basis for $\Sigma^\infty$. If $P$ is a dcpo with an effective basis $B=(b_n)_{n\geq0}$, we say $x \in P$ is \emph{computable} provided 
\begin{equation*}
\{n \in \mathbb{N}|b_n \ll x\}
\end{equation*}
is recursively enumerable \cite{edalat1999domain}. Notice, as in the case of weak bases, this notion is independent of the chosen finite map (see Proposition \ref{model indep bases} $(i)$) and, under a broad hypothesis, also independent of the chosen basis (see Proposition \ref{model indep bases} $(ii)$).


Before introducing computable functions, we need some extra terminology.
We call a map $f:P \rightarrow Q$ between dcpos $P,Q$ a \emph{monotone} if $x \preceq y$ implies $f(x) \preceq f(y)$. Furthermore, we call $f$ \emph{continuous} if it is a monotone and, for any directed set $A \subseteq P$, we have $f(\sqcup A) = \sqcup f(A)$ \cite{abramsky1994domain,scott1970outline}. Note this definition is equivalent to the usual topological definition of continuity (see, for example, \cite{kelley2017general}) applied to the Scott topology. Note, also, if $B \subseteq P$ is a weak basis and a $f$ is continuous, then
$f(P)$ is determined by $f(B)$. This is the case since, if $x \in P$, then there exists some directed set $B_x \subseteq B$ such that $\sqcup B_x=x$. Thus, by continuity of $f$,
$f(x)=f(\sqcup B_x)=\sqcup f(B_x)$. We can consider the monotonicity of $f$ in the definition of continuity as a mere technical requirement to make sure $\sqcup f(A)$ exists for any directed set $A \subseteq P$.

We introduce now computable functions. We say a function $f:P \to Q$ between two dcpos $P$ and $Q$ with effective bases $B=(b_n)_{n\geq0}$ and $B'=(b'_n)_{n\geq0}$, respectively, is \emph{computable} if $f$ is continuous and the set
 \begin{equation}
 \label{compu funct def}
\{\langle n,m \rangle| b'_n \ll f(b_m)\}
 \end{equation}
is recursively enumerable \cite{edalat1999domain}. Hence, we can think of continuity in $\sigma(P)$ as a weak form of computability. This notion of computable function is independent of the finite maps considered for both $P$ and $Q$ (see Proposition \ref{model indep bases} $(i)$) and, under some broad hypotheses, also independent of the chosen bases (see Proposition \ref{model indep bases} $(iii)$).
Because of that, we may specify the chosen bases and call an element $x \in P$ \emph{$B$-computable}, instead of just computable, where $B \subseteq P$ is an effective basis. Analogously, we may call a function $f:P \to Q$ \emph{$(B,B')$-computable}, where $B \subseteq P$ and $B' \subseteq Q$ are effective bases. 
Note, for a fixed pair of effective bases, there are a countable number of computable functions between dcpos, as we show in Proposition \ref{count compu f} in the Appendix \ref{proofs}. In Section \ref{diff comp func} we expose the intuition behind computable functions and argue why they were not defined in the non-uniform approach.


These uniform notions of computable elements and functions can be translated to a set $X$, as in Section \ref{order in compu}, via a partial surjective map $\rho: P \rightarrow X$, which is known as a \emph{domain representation} \cite{stoltenberg2008computability,blanck2008reducibility}. We say $x \in X$ is $\rho$-computable if there exists some computable $p \in P$ such that $x=\rho(p)$ \cite{stoltenberg2008computability}.
Note, as in Section \ref{order in compu}, we are restricted to
sets $X$ with, at most, the cardinality of the continuum
(see Proposition \ref{cardinality where computability}).\footnote{Note 
the cardinality restriction in the basis' case is an instance of the more general topological fact that $T_0$ second countable topological space have, at most, the cardinality the continuum, which can be shown following Proposition \ref{cardinality where computability}.}
Furthermore, given another set $Y$, we say $f:X \rightarrow Y$ is $(\rho,\rho')$-computable if there exists a computable function $g:P \rightarrow Q$ such that $g(\text{dom}(\rho)) \subseteq \text{dom}(\rho')$ and for each $x \in \text{dom}(\rho)$ we have $f(\rho(x))=\rho'(g(x))$ \cite{stoltenberg2008computability}. As they are fundamental in the introduction of computability notions in topological spaces, domain representations have been studied and classified in the literature \cite{stoltenberg2008computability,stoltenberg2001notes,blanck2000domain,blanck2008reducibility,hamrin2005admissible}, emphasizing the case where $X$ is a topological space.
Special attention has been given to domain representations $\rho$ such that $\text{dom}(\rho)=max(P)$ \cite{martin1998domain,waszkiewicz2003domains}, where $max(P) \coloneqq \{x \in P| \nexists y \in P \text{ s.t. } x \prec y \}$.

\section{Comparing the uniform and non-uniform approaches}
\label{comparing apporaches}

The fundamental differences between the non-uniform approach in Section \ref{order in compu} and the uniform one in Section \ref{uniform compu} are the following:

\begin{enumerate}[label=(\roman*)]
\item The substitution of weak bases by bases. More fundamentally, the inclusion of the way-below relation $\ll$ in the uniform approach.
\item The definition of computable elements. If the uniform approach were to follow the non-uniform one using $\ll$, the definition of computable element should require the existence of \emph{some} $B_x \subseteq B \cap \twoheaddownarrow x$ such that $\alpha^{-1}(B_x)$ is recursively enumerable. Instead, it requires $\alpha^{-1}(B \cap \twoheaddownarrow x)$ to be recursively enumerable.
\item The definition of computable functions. Such a notion is absent in the non-uniform approach and present in the uniform one.
\item The inclusion of a bottom element $\perp$. In the uniform case, this is required, while we have avoided it in the non-uniform framework.
\item The definition of effectivity. In the non-uniform approach, we justified the inclusion of effectivity as a tool allowing us to determine what the best approximation to the solution of certain computation is. However, in the uniform case, the requirement is stronger. \end{enumerate}

In this section, we discuss the influence of the differences $(i)$-$(v)$ in terms of three different perspectives: computable elements (Section \ref{diff comp ele}), computable functions (Section \ref{diff comp func}), model dependence of computability (Section \ref{model dep comp}) and complexity (Section \ref{complexity}). 

\subsection{Computable elements}
\label{diff comp ele}

In terms of computable elements, two differences arise between the non-uniform and the uniform approach: the scope, that is, the set of ordered structures where computable elements can be defined in each approach, and the definition of computable elements, where, essentially, the non-uniform approach defines an element as computable if a computation path leading to it can be found, while the uniform apporach defines it as computable if and only if a specif computation path exists. Regarding the scope of these approaches, we show in Proposition \ref{weak basis compu and no basis compu} the non-uniform approach covers a wider range of order structures. Regarding the definition of computable elements, we show in Proposition \ref{compu sets} that, in the uniform approach, the existence of the specific computation path to certain element through which computable elements are defined is equivalent to the existence of any computational path to that element. We conclude, in this section, that the difference in $(i)$ reduces the ordered structures where computation can be introduced, that the difference in $(ii)$ is only apparent and vanishes in the uniform approach and that the inclusion of $(iv)$ and the stronger version of $(v)$ 
in the uniform approach are essential in order for the difference in $(ii)$ to vanish.

We begin, in Proposition \ref{weak basis compu and no basis compu}, showing effective weak basis exist whenever effective bases do. Moreover, we construct a dcpo where bases do not exist and effective weak bases do, which implies the non-uniform approach covers a wider set of ordered structures.

\begin{prop}
\label{weak basis compu and no basis compu}
If $P$ is a dcpo with an effective basis, then it has an effective weak basis. However, there are dcpos with effective weak bases that have no basis.
\end{prop}
\begin{proof}
We begin showing the following straightforward lemma we use in the proof of the first statement.

\begin{lem}
\label{always next}
If $P$ is a dcpo and $A \subseteq P$ is a directed set such that $\sqcup A \notin A$, then for all $a \in A$ there exists some $b \in A$ such that $a \prec b$.
\end{lem}

\begin{proof}
Assume w.l.o.g. $|A| \geq 2$. Take $a_0 \in A$ and assume the result is false. Then $\forall b \in A/\{a_0\}$ either $b \prec a_0$ or $b \bowtie a_0$. If there exists $b$ such that $b \bowtie a_0$ then, since $A$ is directed, there exists $c \in A$ such that $a_0  \preceq c$ and $b \preceq c$. If $c = a_0$ then $b \preceq a_0$, a contradiction. Thus, $a_0 \prec c$, a contradiction. Hence, we must have $a \prec a_0$ $\forall a \in A/\{a_0\}$, which means $\sqcup A = a_0$ since $a \preceq a_0$ $\forall a \in A$ and if $a \preceq y$ $\forall a$ $\in A$ then $a_0 \preceq y$. We get $\sqcup A \in A$, a contradiction.
\end{proof}

We prove now the first statement. Note any basis with a finite map is a weak basis where the same finite map works. Consider the computable function $f$ such that $f(\mathbb{N})=\{\langle n,m \rangle| b_n \ll b_m\}$. Since $x\ll y$ implies $x\preceq y$ $\forall x,y \in P$, we have $f(\mathbb{N}) \subseteq \{\langle n,m \rangle| b_n \preceq b_m\}$. If $x \in P$, then $B_x \coloneqq \twoheaddownarrow x \cap B$ is a directed set such that $\sqcup B_x = x$ \cite[Proposition 2.2.4]{abramsky1994domain}. Take $b_n,b_m \in B_x$ such that $b_n,b_m \prec x$. If $x \in B_x$, then $b_n,b_m \prec x \ll x \eqqcolon b_p$, which means $b_n,b_m \ll b_p$ \cite[Proposition 2.2.2]{abramsky1994domain}. Thus, $\langle n,p \rangle, \langle m,p \rangle \in f(\mathbb{N})$ and we have finished. If $x \not \in B_x$, then there exists some $y \in B_x$ such that $b_n,b_m \prec y$ by Lemma \ref{always next}. Using the interpolation property \cite[Lemma 2.2.15]{abramsky1994domain}, we get some $z \in B_x$ such that $y \ll z \ll x$.
Thus, $b_n,b_m \prec z \eqqcolon b_p$ and $b_n,b_m \ll b_p$, which means $\langle n,p \rangle, \langle m,p \rangle \in f(\mathbb{N})$. This concludes the proof of the first statement.

For the second statement, take $P \coloneqq \big(\big[0,1\big],\preceq\big)$ where
\begin{equation*}
    x \preceq y \iff 
    \begin{cases}
    x \leq y &\text{ if } x,y \in \big[0,\frac{1}{2}\big],\\
    y \leq x & \text{ if } x,y \in \big[\frac{1}{2},1\big].
    \end{cases}
\end{equation*}
Note $B \coloneqq \mathbb{Q} \cap [0,1]$ is a countable weak basis. In fact, $B$ is an effective weak basis. We can easily construct a finite map $\alpha:\mathbb{N} \to \mathbb{Q} \cap [0,1]$ which, aside from $0$ and $1$, orders the rationals in $[0,1]$ lexicographically, considering first the denominators and then the numerators. If $m=0,1$, then $\alpha(m)=m$. If $m>1$, then we
start with $t=2$ and increase $t$ by one unit until we find one such that $m <2+\sum_{i=2}^t (i-1)$. We take then
\begin{equation*}
\alpha(m)=\frac{m-\big(1+\sum_{i=2}^{t-1} (i-1)\big)}{t}.
\end{equation*}
Thus, $\alpha$ is a finite map. Note we can show similarly that $\{\langle n,m\rangle| \alpha(n) \preceq \alpha(m)\}$ is recursively enumerable.
However, $P$ has no countable basis. This is the case since there is no $x \in P$ such that $x \ll \frac{1}{2}$. If $x \in \big[0,\frac{1}{2}\big]$, then $D_0 \coloneqq \mathbb{Q} \cap \big(\frac{1}{2},1\big]$ is a directed set such that $\sqcup D_0= \frac{1}{2}$ and $d \bowtie x$ $\forall d \in D_0$. Thus, $\neg\big(x \ll \frac{1}{2}\big)$. If $x \in \big(\frac{1}{2},1\big]$, we can argue analogously, taking $D_1 \coloneqq \mathbb{Q} \cap \big[0,\frac{1}{2}\big)$.
\end{proof}

Note Proposition \ref{weak basis compu and no basis compu} clarifies why Definition \ref{def:eff weak basis} cannot be simplified, as it seems likely, although it remains a question, that effective bases such that $\{\langle n,m \rangle| b_n \preceq b_m\}$ is not recursively enumerable exist.

Although the way in which we introduced computable elements through weak bases and bases seems to differ, whenever bases exist, both definitions are equivalent, as we show in Proposition \ref{compu sets}. In particular, we show that the definition of computable elements in the uniform approach is equivalent to the one in the non-uniform framework, provided we have that $B$ in Definition \ref{def:compu ele} is, in particular, an effective basis.  

\begin{prop}
\label{compu sets}
If $P$ is a dcpo with an effective basis $B=(b_n)_{n\geq0}$ and $x \in P$, then the following are equivalent:
\begin{enumerate}[label=(\roman*)]
\item There exists a directed set $B_x \subseteq B$ such that $\sqcup B_x=x$ and $\{n\in\mathbb{N}|b_n \in B_x\}$ is recursively enumerable.
\item $x$ is $B$-computable.
\end{enumerate}
\end{prop}

\begin{proof}
Note $(ii)$ implies $(i)$ by definition of computable elements in the basis sense, since $B_x \coloneqq \twoheaddownarrow{x} \cap B$ is a directed set with $\sqcup B_x=x$ \cite[Proposition 2.2.4]{abramsky1994domain} and, given that $x$ is computable, $\{n\in\mathbb{N}|b_n \in B_x\}$ is recursively enumerable by definition.

For the converse, note, if we have a directed set $B_x \subseteq B$ such that $\sqcup B_x=x$, then it follows that
\begin{equation*}
    \{n\in\mathbb{N}|b_n \ll b_m, \text{ } b_m \in B_x\} = \{n\in \mathbb{N}|b_n \ll x\}.
\end{equation*}
$(\subseteq)$ holds by \cite[Proposition 2.2.2]{abramsky1994domain}, since
we have $b_n \ll b_m \preceq x$, hence, $b_n \ll x$. $(\supseteq)$ also holds, since we have $b_n \ll x$ and $\sqcup B_x=x$. Thus, by \cite[Corollary 2.2.16]{abramsky1994domain}, there exists some $b_m \in B_x$ such that $b_n \ll b_m$. To conclude the proof, we only need to show $\{n \in \mathbb{N}| b_n \ll x\}$, or, equivalently, $\{n\in\mathbb{N}|b_n \ll b_m, \text{ } b_m \in B_x\}$, is recursively enumerable using the fact that both $\{n\in\mathbb{N}|b_n \in B_x\}$ and $\{\langle n,m \rangle| b_n \ll b_m\}$ are. In order to do so, we construct a computable function $i: \mathbb{N} \to \mathbb{N}$ (see Figure \ref{fig 2} for a representationof $i$) such that $i(\mathbb{N})=\{n\in\mathbb{N}|b_n \ll b_m, \text{ } b_m \in B_x\}$ from two computable functions, which we know exist, $h,g:\mathbb{N} \to \mathbb{N}$, where $g(\mathbb{N})=\{n\in\mathbb{N}|b_n \in B_x\}$ and $h(\mathbb{N})=\{\langle n,m\rangle| b_n \ll b_m\}$. We build, thus, an algorithm which goes through all the outputs of $h$ and $g$ and checks whether the first component of $h$ equals an output of $g$. In case it does, it outputs the first component of $h$ and, otherwise, it outputs $0$, since $b_0 =\perp \ll x$ $\forall x \in P$. Consider, then, some $n \in \mathbb{N}$. We start with $t=1$ and continue increasing $t$ one unit until we have $n < \sum_{m=1}^t (2m-1)$. $i$ has a different output, then, according to an extra comparison. If $n \leq t+\sum_{m=1}^{t-1} (2m-1)$, then we compare $h(t)$ with $g(t + \sum_{m=1}^{t-1} (2m-1)-n)$, that is, in case $\pi_2 \circ h(t)=g(t + \sum_{m=1}^{t-1} (2m-1)-n)$, then we output $i(n)=\pi_1 \circ h(t)$ and, otherwise, we output $i(n)=0$. Conversely, if $n > t+\sum_{m=1}^{t-1} (2m-1)$, then we compare $g(t)$ with $h(n-(1+t + \sum_{m=1}^{t-1} (2m-1)))$, that is, in case $\pi_2 \circ h(n-(1+t + \sum_{m=1}^{t-1} (2m-1)))=g(t)$, we output $i(n)=\pi_2 \circ h(n-(1+t + \sum_{m=1}^{t-1} (2m-1)))$ and, otherwise, we output $i(n)=0$. Note $i$ is computable, since $g,h,\pi_1 $ and $\pi_2$ are.
\end{proof}

Note the stronger version of effectivity in the uniform approach $(v)$ and the inclusion of a bottom element $(iv)$ are essential for the equivalence in Proposition \ref{compu sets}. The first one, $(v)$, assures the existence of any computational path and that of the specific one required in the uniform approach are equivalent. Note a direct extension of effectivity to the uniform case would fail to have this property, since it is precisely the stronger version used in the uniform approach what allows to connect any computational path in the basis to the specific one used to define computability in that framework. The second one, $(iv)$, allows us to construct a computable function for such a specific computational path whenever such a function exist for any path. This is the case since such a construction is based on comparisons and we need, whenever the comparisons fail to be true, some output which is way below any element in the order structure. Note this is a reason for fixing the bottom element as the first element in the enumeration of any basis.

\begin{figure}[!tb]
\centering
\begin{tikzpicture}
     \node[main node] (2) {$h(0)$};
    \node[main node] (4) [right = 2cm  of 2]  {$h(1)$};
    \node[main node] (6) [right = 2cm  of 4]  {$h(2)$};
    \node[main node] (1) [below = 2cm  of 2] {$g(0)$};
    \node[main node] (3) [right = 2cm  of 1] {$g(1)$};
    \node[main node] (5) [right = 2cm  of 3] {$g(2)$};

   \path[draw,thick,-]
    (1) edge node {} (2)
    (1) edge node {} (4)
    (3) edge node {} (2)
    (3) edge node {} (4)
    ;
    
    \path[draw,thick,-,dotted]
    (1) edge node {} (6)
    (3) edge node {} (6)
    (5) edge node {} (2)
    (5) edge node {} (4)
    (5) edge node {} (6)
    ;
\end{tikzpicture}
\caption{Representation of $i:\mathbb{N} \to \mathbb{N}$ from Proposition \ref{compu sets} for $n=1,..,9$. Each line stands for a comparison done by $i$ between an output of $g$ and one of $h$. The continuous lines represent the comparisons in $i(n)$ for $n\leq4$ and the discontinuous those for $5\leq n \leq 9$.}
\label{fig 2}
\end{figure}
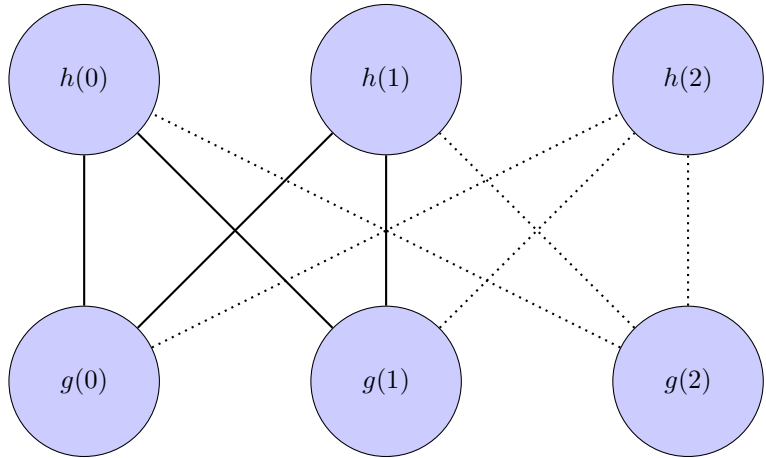

\subsection{Computable functions}
\label{diff comp func}
Regarding computable functions, the difference between the uniform and non-uniform approaches is we defined them in the former and did not in the latter. The reason for this lies in the intuition any such definition aims to capture. In particular, computable functions are meant to map
computable inputs to computable ouputs \cite{ko2012complexity,braverman2005complexity}. We show, in Proposition \ref{why << needed}, the definition in the uniform approach fulfills this property and discuss, right after, why it seems that the non-uniform approach fails to do so. We conclude, in this section, that the difference in $(i)$, i.e. the inclusion of $\ll$, and the inclusion of a bottom element in $(iv)$ are fundamental in order to have a definition of computable function that captures the intuition behind it. 

We begin, in Proposition \ref{why << needed}, showing the definition of computable functions in the uniform approach fulfills the intuitive notion. 

\begin{prop}[\cite{edalat1999domain}]
\label{why << needed}
Let $P$ and $Q$ be dcpos
with effective bases $B$ and $B'$, respectively. If $f:P \to Q$ is a $(B,B')$-computable function and $x \in P$ is computable, then $f(x)$ is computable.
\end{prop}

\begin{proof}
Let $B=(b_n)_{n\geq0}$ and $B'=(b'_n)_{n\geq0}$ be, respectively, effective bases for $P$ and $Q$. Given some computable $x \in P$, we will show $f(x)$ is also computable, that is, that there exists a directed set $B'_{f(x)} \subseteq B'$ such that $\sqcup B'_{f(x)} = f(x)$ and $\beta^{-1}(B'_{f(x)})$ is recursively enumerable, where
$\beta$ is a finite map for $B'$.
Take $B_x \coloneqq \{b_n \in B| b_n \ll x\}$ and define
$B'_{f(x)} \coloneqq \{b'_n \in B'| \exists  b_m \in B_x \text{ s.t. } b'_n \ll f(b_m)\}$.
To begin, we show $B'_{f(x)}$ is a directed set.
If $b'_n,b'_m \in B'_{f(x)}$, then there exist $b_n,b_m \in B_x$ such that $b'_n \ll f(b_n)$ and $b'_m \ll f(b_m)$. Since $B_x$ is directed, there exists some $b_p \in B_x$ such that $b_n,b_m\preceq b_p$ and, as $f$ is monotonic, $f(b_n),f(b_p)\preceq f(b_p)$. Thus, by \cite[Proposition 2.2.2]{abramsky1994domain}, we have $b'_n,b'_m \ll f(b_p)$ and, by \cite[Lemma 2.2.15]{abramsky1994domain}, there exist some $b'_p \in B'$ such that $b'_n,b'_m \ll b'_p \ll f(b_p)$. Since $b'_p \in B'_{f(x)}$ by definition, we conclude $B'_{f(x)}$ is directed and, as $Q$ is a dcpo, $\sqcup B'_{f(x)}$ exists.
To conclude, we show $\sqcup B'_{f(x)}=f(x)$.
Notice, by definition of $ B'_{f(x)}$, $\sqcup B'_{f(x)} \preceq f(x)$. To prove the converse, notice, since $B'$ is a basis for $Q$, there exists for each $b_n \in B_x$ some $B'_n \subseteq B'_{f(x)}$ such that $\sqcup B'_{n}=f(b_n)$. Thus, we have $f(b_n) \preceq \sqcup B'_{f(x)}$ $\forall b_n \in B_x$. Given that $f$ is continuous, it follows $\sqcup f(B_x) = f(x)$ and, by definition of supremum, that $f(x) \preceq \sqcup B'_{f(x)}$. We obtain $\sqcup B'_{f(x)} = f(x)$ by antisymmetry. Lastly, we can follow Proposition \ref{compu sets} to show $\beta^{-1}(B'_{f(x)})$ is recursively enumerable, taking, in the notation of Proposition \ref{compu sets}, $h,g:\mathbb{N} \to \mathbb{N}$ such that $g(\mathbb{N})=\{n | b_n \ll x\}$ and $h(\mathbb{N})=\{\langle n,m\rangle| b'_n \ll f(b_m)\}$.
\end{proof}

Note our approach in the proof of Proposition \ref{why << needed} is more direct than the one in \cite[Theorem 9]{edalat1999domain}, where the converse of the statement is also addressed.
Note, also, we have averted using \emph{oracles}, which appear in the literature to conveniently decouple the difficulty of computing the input from that of the output and to avoid limiting the input of computable functions to computable elements \cite{braverman2005complexity}.

What the intuition behind computable functions aims to capture is that, if there exists an effective procedure for some input and a computable function $f$, there ought to be an effective procedure to determine the output. In particular, one is able to construct such a procedure for $f(x)$ when given such a procedure for some element $x$. This holds in the uniform approach, as we showed in Proposition \ref{why << needed}. When trying to define computable functions using weak bases, however, this is not necessarily true. To support this claim informally, we introduce an example in the following.

Take two partial orders $P,Q$ such that $P \subseteq Q$ and the identity map between them $f:P \to  Q$. Whatever definition of computable function one chooses, one would like the identity map to be computable.
The problem rises when one attempts to provide a definition in terms of order properties that respects this fact. Assume, for the moment, that $P$ and $Q$ have effective bases $B$ and $B'$, respectively. Then, according to \eqref{compu funct def}, the computability of $f$ reduces to whether $\{\langle n,m\rangle| b'_n \ll_Q b_m\}$ is recursively enumerable. Note it is reasonable to expect this to be true, given that $\{\langle n,m\rangle| b'_n \ll_Q b'_m\}$ is recursively enumerable since $B'$ is an effective basis and there exists a finite map $\alpha: \mathbb N \to B$. The situation for weak bases is quite different, as the following example points out.

Consider the same situation from the previous paragraph, where we take the identity function $f: P \to Q$. Fix, in particular, $P \coloneqq (\mathbb{N} \cup \{\infty\},\preceq_P)$ and $Q \coloneqq ([0,\infty] \cup \mathbb{Z}^{-},\preceq_Q)$, where we have $x \preceq_P y \iff y=\infty$ or $x \leq y$ with $x,y \in \mathbb{N}$ and
\begin{equation}
\label{def Q 2}
    \begin{rcases}
     y=\infty, \\
    x \leq y &\text{ if } x,y \in \mathbb{Z}^+, \\
    y \leq x &\text{ if } x,y \in \mathbb{Z}^{-}, \text{ or }\\
    x \leq y &\text{ if } x,y \in (n,n+1] \text{ for some } n \in \mathbb{N} 
    \end{rcases}
    \implies x \preceq_Q y.
\end{equation}
The definition of $\preceq_Q$ is then completed by adding the relations implied by transitivity and \eqref{def Q 2}. Note that $B=P\setminus \{x\}$ and $B'=(Q \cap \mathbb{Q})\setminus(\mathbb{Z}^+\setminus\{0\})$ are countable weak bases for $P$ and $Q$, respectively, which can be shown to be effective. Any definition of computable functions that resembles the one in \eqref{compu funct def} using weak bases will end up relating the elements in $f(B)$ with the elements in $B'$ that are below below them according to $\preceq_Q$, as one does via $\ll$ in the basis case \eqref{compu funct def}. Hence, if we fix the computable element $x = \infty \in P$ in our example, the output of any effective procedure that can be derived from $x$ and $f$ to compute $f(x)=\infty \in Q$ will be contained inside $ B' \cap [0,\infty)$. However, given that there is no directed set $A_\infty \subseteq B' \cap [0,\infty)$ such that $\sqcup A_\infty = \infty$, no effective procedure to compute $f(x)$ can be derived from $f$ and $x$ alone. Nonetheless, $f(x)$ is actually computable in $Q$, since it is the upper bound of $\mathbb Z^{-}$. The issue is that any computational path leading to $f(x)$ in $Q$ has no connection via $f$ and $\preceq_Q$ to the paths leading to $x$ in $P$. As a result, in the non-uniform case, there does not seem to be a definition of computable functions that fulfills our intuition and includes the identity function in general. The situation differs from that in the uniform case since the appeal to $\ll$ prevents the appearance of disconnections. In fact, note that the fact $Q$ has no basis is key in our example (this is easy to see given that there is no $y \in Q$ such that $y \ll_Q \infty$).

In summary, we need some structure aside from $\preceq$, like $\ll$, to define computable functions. This is the case since, in the non-uniform approach, the lack of extra structure provided by $\ll$ can result in a disconnection between the computable paths that lead to $x$ and to $f(x)$. In conclusion, the difference in $(i)$ seems to be fundamental for the introduction of computable functions. Note the inclusion of $(iv)$ fulfills the same purpose in Propositions \ref{compu sets} and \ref{why << needed}.

\subsection{Model dependence of computability}
\label{model dep comp}

In this section, we study how computability of elements and functions depends on the structure used to define it in both the uniform and the non-uniform approach. We conclude that computable elements are independent of the chosen finite map in both approaches and that, in the non-uniform approach, the same holds true for computable functions. Regarding the dependence on the chosen weak basis/basis, we discuss why computable elements seem to depend on the weak basis in the non-uniform approach and give general conditions under which both computable elements and functions are independent of the chosen basis in the uniform approach. This difference between both approaches supports the inclusion of restrictions $(i)$ and $(iv)$ as well. 

We begin, in Proposition \ref{indep finite map}, showing computable elements are independent of finite maps and, right after, discuss informally why this is not the case for weak bases. 

\begin{prop}
\label{indep finite map}
If $P$ is a dcpo with an effective weak basis, then the set of computable elements is independent of the chosen finite map.
\end{prop}

\begin{proof}
Assume $\alpha$ and $\beta$ are finite maps for an effective weak basis $B$. In order to get the result, it is sufficient to show, whenever $\alpha^{-1}(B_x)$ is recursively enumerable for some directed set $B_x \subseteq B$, then $\beta^{-1}(B_x)$ is also recursively enumerable. Since $\alpha$ and $\beta$ are finite maps, we have $\beta^{-1} \circ \alpha$ is an effectively computable function in the informal sense. Thus, by Church's thesis \cite{cutland1980computability,rogers1987theory}, there exists a partially computable function, in the formal sense, $h: \mathbb{N} \to \mathbb{N}$ such that $h(n)=\beta^{-1} \circ \alpha(n)$ $\forall n \in \mathbb{N}$. On the other hand, since $\alpha^{-1}(B_x)$ is a recursively enumerable set, then, by definition, there exists a computable function $f:\mathbb{N} \to \mathbb{N}$ such that $f(\mathbb{N})=\alpha^{-1}(B_x)$. As a result, $h \circ f$ is a computable function such that $h \circ f(\mathbb{N}) = \beta^{-1}(B_x)$, that is, $\beta^{-1}(B_x)$ is recursively enumerable.
\end{proof}

We argue now, informally, why computable elements seem to depend on the chosen weak basis. Take $Q \coloneqq ((-\infty,\infty],\preceq_Q)$, where $\preceq_Q$ consists of the relations in \eqref{def Q 2}, plus $x \preceq_Q y$ if $y \leq x$ with $x,y \in [-(n+1),-n)$ for some $n \geq 0$, and the ones derived from them by transitivity. Note that both $B=(Q \cap \mathbb{Q})\setminus\mathbb{Z}^-$ and $B'=(Q \cap\mathbb{Q})\setminus\mathbb{N}$ can be shown to be effective weak bases for $Q$. Moreover, note that any condition we may impose to prove that computability in the $B$ and $B'$ senses are equivalent would rely on some effective procedure to determine the relation $\preceq$ between the elements in $B$ and $B'$. This would allow us, thus, to find an effective procedure for any basis when such a procedure exists for the other. However, as in the case of the definition of computable functions, the disconnection that is allowed by weak bases and prevented by bases renders this false in our example.
More specifically, this does not hold true for $x=\infty \in Q$ since, although $B_x = B \cap \mathbb{N}$ is directed and fulfills $\sqcup B_x=x$, 
there is no directed set $B'_x \subseteq B'$ such that $\sqcup B'_x=x$ and whose elements are related to the ones in $B_x$ by $\preceq$.
Nonetheless, $x$ is of course computable in the $B'$ sense, since we have that $B'_x = B' \cap \mathbb{Z}^-$ fulfills $\sqcup B'_x=x$. However, the point is that such a procedure is completely unrelated in terms of $\preceq$ to $B_x$.
Thus, it is not unreasonable to assume there will be some dcpo where such a disconnection between two bases will lead to some element being computable in one scenario and not in the other one. Regarding the situation in the uniform approach, as we will show in Proposition \ref{model indep bases}, this lack of connection does not show up and, hence, computability is independent of the chosen basis provided a certain condition is fulfilled.

We consider now model dependence of computability in the uniform approach. In particular, in Proposition \ref{model indep bases}, we show both computable elements and functions are independent of the finite maps involved and, under broad hypotheses, also independent of the chosen effective bases.

\begin{prop}
\label{model indep bases}
If $P$ and $Q$ are dcpos with effective bases, then the following statements hold:
\begin{enumerate}[label=(\roman*)]
\item The set of computable elements in $P$ is independent of the chosen finite map. Moreover, the set of computable functions $f:P \to Q$ is independent of the chosen finite maps.
\item If $B=(b_n)_{n\geq0}$ and $B'=(b'_n)_{n\geq0}$ are effective bases for $P$ such that $\{\langle n,m\rangle| b'_n \ll b_m\} $ is recursively enumerable, then $x \in P$ is $B'$-computable provided it is $B$-computable.
\item Let $B=(b_n)_{n\geq0}$ and $B'=(b'_n)_{n\geq0}$ be effective bases for $P$ and both $C=(c_n)_{n\geq0}$ and $C'=(c'_n)_{n\geq0}$ be effective bases for $Q$. If both $\{ \langle n,m \rangle | b_n \ll b'_m\}$ and $\{ \langle n,m \rangle | c'_n \ll c_m\}$ are recursively enumerable, then any $(B,C)$-computable function $f:P \to Q$ is $(B',C')$-computable.
\end{enumerate}
\end{prop}

\begin{proof}
$(i)$ To show the first statement, we can slightly modify the proof of Proposition \ref{indep finite map}. We finish showing the result for computable functions. Assume we
$B\subseteq P$ and $B'\subseteq Q$ are effective bases, $\alpha,\beta: \mathbb{N} \to B$ and $\alpha',\beta': \mathbb{N} \to B'$ are finite maps  and $f:P \to Q$ is a $(\alpha,\alpha')$-computable function. We will show $f$ is $(\beta,\beta')$-computable. As in the proof of Proposition \ref{indep finite map}, by Church's thesis, there exist partially computable functions $g,h:\subseteq \mathbb{N} \to \mathbb{N}$ such that $g= \beta^{-1} \circ \alpha$ and $h= (\beta')^{-1} \circ \alpha'$. Thus, given a computable function $i:\mathbb{N} \to \mathbb{N}$ such that $i(\mathbb{N})=\{\langle n,m \rangle| \alpha(n) \ll f(\alpha'(m))\}$, we have $j: \mathbb{N} \to \mathbb{N}$, $n \mapsto \langle g \circ \pi_1 \circ i(n), h \circ \pi_2 \circ i(n) \rangle$ is a computable function, since $g,h,\pi_1,\pi_2$ and $i$ are, and, by definition of $g$ and $h$, $j(\mathbb{N})=\{ \langle n,m \rangle| \beta(n) \ll f(\beta'(m)) \}$.

$(ii)$ If $x \in P$ is $B$-computable, then there exists a computable function $f: \mathbb{N} \to \mathbb{N}$ such that $f(\mathbb{N})=\{n | b_n \ll x\}$. Note there is another computable function $j: \mathbb{N} \to \mathbb{N}$ such that $j(\mathbb{N})=\{\langle n,m\rangle| b'_n \ll b_m\}$ by hypothesis and, given that $b'_m \in B'_x = \twoheaddownarrow x \cap B'$, there exists some $b_m \in B_x$ such that $b'_m \ll b_m \ll x$ by the interpolation property \cite{abramsky1994domain}. We can follow the proof of Proposition \ref{compu sets} to construct a computable function $i: \mathbb{N} \to \mathbb{N}$ such that $i(\mathbb{N})=\{n | b'_n \ll x\}$. In particular,  we take, in the notation of Proposition \ref{compu sets}, $g \coloneqq f$ and $h \coloneqq j$.
Note we have $i(\mathbb{N})=\{n|b'_n \ll x\}$ as a consequence of the interpolation property.

$(iii)$ We will construct a computable function $i: \mathbb{N} \to \mathbb{N}$ such that $i(\mathbb{N})=\{ \langle n,m \rangle | c'_n \ll f(b'_m)\}$. Notice, by assumption, there exist computable functions $g,h,j: \mathbb{N} \to \mathbb{N}$ such that $g(\mathbb{N})=\{ \langle n,m \rangle| c'_n \ll c_m\}$, $h(\mathbb{N})=\{ \langle n,m \rangle| c_n \ll f(b_m)\}$ and $j(\mathbb{N})=\{ \langle n,m \rangle| b_n \ll b'_m\}$. Notice, also, if $c'_n \ll f(b'_m)$ for some $n,m \in \mathbb{N}$, then there exists some $p \in \mathbb{N}$ such that $c'_n \ll c_p \ll f(b'_m)$ by the interpolation property. Since $B=(b_n)_{n\geq0}$ is a basis for $P$ and $f$ is continuous in the Scott topology, we have $\sqcup \{ f(b_n)| b_n \ll b'_m\} = f(b'_m)$ $\forall m \in \mathbb{N}$. Thus, whenever $c_p \ll f(b'_m)$ for some $p \in \mathbb{N}$, there exists some $q \in \mathbb{N}$ such that $c_p \ll f(b_q)$ by \cite[Corollary 2.2.16]{abramsky1994domain}. Because of the previous observations, we can now construct $i$ similarly
to how we constructed it in the proof of Proposition \ref{compu sets}, only, in this case, we need to compare the output of three functions (instead of two) in order to determine the output of $i$. We can fix the comparison procedure in a similar way to how we fixed it in Proposition \ref{compu sets}. Once this is done, we get for each $n \in \mathbb{N}$ a tuple $(r_n,s_n,t_n) \in \mathbb{N}^3$ and will output $i(n) \coloneqq \langle \pi_1 \circ g(r_n),\pi_2 \circ j(t_n) \rangle$ in case we have both $\pi_2 \circ g(r_n) = \pi_1 \circ h(s_n)$ and $\pi_2 \circ h(s_n) = \pi_1 \circ j(t_n)$ and, if any of these is false, $i(n) \coloneqq 0$.
\end{proof}

Note that, informally, the hypotheses in Proposition \ref{model indep bases} $(ii)$ and $(iii)$ are likely to be true. This is the case since, as the involved basis elements have a finite representation given, respectively, by the finite maps and $\ll$ is recursively enumerable by hypothesis when considered between elements of the same basis, it is reasonable to expect $\ll$ to be recursively enumerable when comparing pairs of elements in different bases. Note, also, the proof of Proposition \ref{model indep bases} $(ii)$ and $(iii)$ supports restrictions $(i)$ and $(iv)$ as well. The first one, $(i)$, is supported since, as we discussed right above Proposition \ref{model indep bases}, there does not seem to be a hypothesis under which computability is independent of the weak basis in the non-uniform approach. On the contrary, we give a general hypothesis, which is likely to be true, that supports model independence in the non-uniform approach. The second one, $(iv)$, is supported by Proposition \ref{model indep bases} along similar lines than Proposition \ref{why << needed}, since it provides a correct output when some comparisons fail to be true.

\subsection{Complexity}
\label{complexity}

In this section, we study how complexity of elements and functions can be introduced in both the non-uniform and the uniform approaches. We conclude that complexity can be defined for both in the uniform approach and that, although computable elements are defined in the non-uniform framework, there does not seem to be a suitable complexity notion for them. Notice, to our best knowledge, complexity has not been previously addressed in the uniform framework, that is, in domain theory.


We follow closely the approach to complexity for real numbers and real-valued functions in \cite{ko1998polynomial,ko2012complexity} to develop complexity for the
uniform order-theoretic approach. There, finite maps $\alpha:\mathbb{N} \to D$ are considered, where $D$ is the set of dyadic rationals, that is, the rational numbers with a finite binary expansion.
Using these, a real number $x \in \mathbb{R}$ is said to be \emph{computable} if there exists a computable function $\phi: \mathbb{N} \rightarrow \mathbb{N}$ such that for all $n \geq 0$ we have $|\alpha (\phi(n))-x| \leq 2^{-n}$. Moreover, the time (space) complexity of $x$ is said to be bounded by $t:\mathbb{N} \rightarrow \mathbb{N}$ if there exists a computable function $\phi: \mathbb{N} \rightarrow \mathbb{N}$ which,
for any $n \in \mathbb{N}$, calculates $\phi(n)$ in $t(n)$ steps (using $t(n)$ memory cells) and $|\alpha(\phi(n))-x|\leq 2^{-n}$.

This approach to computable real numbers requires a computational procedure to enumerate a subset of the dyadic rationals which converges to a real number \emph{with certain convergence rate}. This is done so in order to be able to extract information about the result of some process by examining its outputs as, if  no restriction on the convergence rate exists, then a finite subset of the dyadic rationals could converge to any real number \cite{di1996real}. Thus, the condition on the convergence rate is needed to obtain information regarding the direction of some computational process by examining a finite subset of its outputs, the same intuitive reason why we introduced effectivity for weak bases. However, the convergence rate is needed to measure complexity since, even assuming we can determine the direction of the process in some other way, the complexity of a real number is supposed to capture the idea of how hard it is to extract information about it and, as such, certain measure of how much information about it the partial outputs of the process have is needed. Because of this, the reason why a distance function appears in the definition of computability and complexity in \cite{ko1998polynomial,ko2012complexity} is
slightly different. In the first one, it is used to assure the process converges to the desired real number $x \in \mathbb{R}$, while, in the second one, it measures how close certain output is from $x$. This is key to define complexity, as it aims to reflect the amount of resources needed to achieve certain level of precise knowledge about $x$. This difference surfaces in the uniform order-theoretic approach, where, although the distance functions is irrelevant to assure convergence and, thus, to define computability (see Definition \ref{def:compu ele}), a closely related notion is indeed needed in order to define complexity (see Definition \ref{def:element complex}).   

The main ingredients used in the definition of computable real numbers and complexity are: (a) for any $x$, the existence of a recursive map $\phi$ whose image under a finite map $\alpha$ converges to $x$ and (b) the existence of a metric $d$, here $d(x,y)=|x-y|$, a way of measuring precision. Note that we specifically use $f(n)=2^{-n}$ as a bound on precision in the definition of computable real numbers, although we could have used any $f:\mathbb{N} \rightarrow \mathbb{R}$ where $\lim_{n \rightarrow \infty} f(n)=0$ holds. Note, also, we can substitute $D$ by $\mathbb{Q}$ \cite{ko1998polynomial}.

The definition of an element $x$ being computable in the uniform order-theoretic approach takes care of (a), since we require the existence of a computable function $\phi:\mathbb{N} \to \mathbb{N}$ whose output, translated via a finite map $\alpha: \mathbb{N} \to B$, is the set of elements in some countable basis $B$ which are way-below $x$ and, thus, converge to $x$ in the Scott topology. Dealing with (b) is not exactly the same since, if a distance function exists, the Scott topology is Hausdorff and $\preceq$ becomes trivial by \eqref{charac order by topo}, as for any pair $x,y \in P$ $x \neq y$ there would be some $O \in \sigma(P)$ such that $x \in O$ and $y \not \in O$. As we discussed in Section \ref{order in compu}, this prevents us from introducing computability to uncountable spaces, since we end up with in the theory of numberings \cite{ershov1999theory}. Given that we cannot avoid a notion of precision to define of complexity, we use a definition from domain theory to fill in this gap: maps inducing the Scott topology everywhere \cite{martin2000foundation,martin2008technique,waszkiewicz2001distance,waszkiewicz2003quantitative}. A monotone map $\mu:P \rightarrow [0,\infty)^{op}$, where $[0,\infty)^{op}$ denotes the dcpo composed by the non-negative real numbers equipped with \emph{reversed order} $\leq_{op}$, that is, $x \leq_{op} y$ if and only if $x \geq y$ for all $x,y \in [0,\infty)$,  is said to \emph{induce the Scott topology everywhere in $P$} if, for all $O \in \sigma(P)$ and  $x \in P$, there exists some $\varepsilon >0$, whenever $x \in O$, such that $x \in \mu_{\varepsilon}(x) \subseteq O$, where
\begin{equation*}
    \mu_\varepsilon(x) \coloneqq \{ y \in P| y \preceq x \text{ and } \mu(y)< \varepsilon\}.
\end{equation*}

Notice the main idea behind $\mu$ is to have a magnitude that quantifies the amount of information a computational process has gathered about what it is intending to calculate. Since, in our picture, any element is characterized in terms of the Scott open sets it belongs to, such a magnitude is intimately related to them. To follow previous approaches through distance functions, we need, in particular, continuous maps that induce the Scott topology everywhere to define complexity. 
In summary, maps inducing the Scott topology everywhere substitute distance functions in our approach. This is the case since any element can be characterized through all the open sets in the Scott topology it belongs to, given that the Scott topology is $T_0$, and approaching $\mu(x)$ in any computational process leading to $x$ implies the outputs of the process posses more information about $x$.

Now that we have introduced the substitute for distance functions in the order-theoretic setups, we can define complexity of both elements and functions in the uniform approach. Note, after dealing with the uniform approach, we will address complexity in the non-uniform one.

\begin{defi}[Element complexity]
\label{def:element complex}
Take $P$ a dcpo with an effective basis $(b_n)_{n\geq0}$ and $\mu$ a continuous map inducing the Scott topology everywhere in $P$. We say the time (space) complexity of an element $x \in P$ is bounded by $t:\mathbb{N} \rightarrow \mathbb{N}$ if there exists a computable function $\phi:\mathbb{N} \to \mathbb{N}$ which computes $\phi(n)$ in $t(n)$ steps (using $t(n)$ cells)  for all $n \in \mathbb{N}$
and we have both
$\mu(b_{\phi(n)})-\mu(x) < 2^{-n}$ $\forall n \geq 0$ and $\sqcup (b_{\phi(n)})_{n\geq 0}=x$.
\end{defi}


As an example where complexity can be introduced following Definition \ref{def:element complex}, we can consider the dcpo $P= [0,1]$, with the usual ordering in the real numbers $\leq$,
$\mu: P \to [0,\infty)^{op}$ where $x \mapsto 1-x$, and $(b_n)_{n\geq0}$ is the numeration of the rationals in $P\backslash\{1\}$ that starts with $0$ and orders the rest lexicographically, that is, it begins with the ones with lower denominator and orders those with the same denominator starting with those with lower numerator. Note that $P$ introduces a notion of computable real numbers, since $(b_n)_{n\geq0}=\mathbb{Q}$ is an effective basis and we have, thus, that $x$ is computable if $\{n \in \mathbb{N}|b_n < x\}$ is recursively enumerable. This is the case since one can directly see that $q \ll x$ if and only if $q < x$ for all $(q,x) \in \mathbb{Q} \times \mathbb{R}$. Furthermore, because of $\mu$, we can also introduce complexity on the real numbers through $P$. For example, in this setup, the time complexity of $x \coloneqq 1$ is bounded by $t_0:\mathbb{N} \rightarrow \mathbb{N}$, where $0 \mapsto 1$ and, if $n \geq 1$, $n \mapsto 2^{n+1}-2$. To see this, take the map $\phi_0: \mathbb N \to \mathbb N$, where $0 \mapsto 1$ and, if $n \geq 1$, $n \mapsto 1 + \sum^{2^{n+1}-1}_{i=2} (i-1)$. To conclude, note that $\phi_0(n)$ is computed in $t_0(n)$ steps for each $n \in \mathbb N$, that $\mu(b_{\phi_0(n)})-\mu(x)=2^{-(n+1)} < 2^{-n}$ $\forall n \geq 0$, and that $\sqcup (b_{\phi_0(n)})_{n\geq 0}=x$. Hence, the time complexity of $x$ is bounded by $t_0$.

As a more interesting example, we can, analogously, bound the complexity of $x \coloneqq \pi$ using $P \coloneqq [3,4]$ with its usual ordering and the same $(b_n)_{n \geq 0}$ from the last example but adding $3$ to each element ($\mu$ is also changed accordingly). Since giving explicit formulas is somewhat tedious in this case, we simply expose the general idea. The basic construction starts with an increasing sequence of rational numbers that converges to $x$, $(a_n)_{n \geq 0} \subseteq P \cap (b_n)_{n \geq 0}$. For example, we can follow Leibniz's formula and take $a_n \coloneqq 8 \sum_{k=0}^{n} ((4k+1)(4k+3))^{-1}$ for all $n \geq 0$. The next step is to profit from the regularity of $(a_n)_{n \geq 0}$ to find a map $\phi_\pi: \mathbb N \to \mathbb N$ with the properties required in the definition. This may involve iterating over a map $\phi'_\pi: \mathbb N \to \mathbb N$ that exhaustively enumerates the indices associated to $(a_n)_{n \geq 0}$ as members of $(b_n)_{n \geq 0}$ and following the iteration until the desired distance, in terms of $\mu$, is achieved (this can be evaluated by using a rational upper bound of $\pi$). To conclude, we can use $\phi_\pi$ to determine a bound $t_\pi$ on the complexity of $\pi$.

Note that continuous maps inducing the Scott topology everywhere exist for all dcpos with effective bases (see Proposition \ref{existence map inducing} in the Appendix \ref{inducing}). The basic idea behind Definition \ref{def:element complex} is that any $x \in P$ is characterized by the open sets (under the Scott topology) it belongs to. Thus, $\mu$ measures the information about $x$ at each stage of some computational process that converges to $x$.
Note we use Proposition \ref{compu sets} in Definition \ref{def:element complex}, since we ask for certain computational process converging to $x$ instead of limiting ourselves to the computational process through which $x$ is defined in the uniform approach. Note, also, we used again $f(n)=2^{-n}$ to bound precision by analogy with the case of the real numbers above, although we could, alternatively, have used some other function as we mentioned above. Note the continuity of $\mu$ is key in order for $(\mu(b_{\phi(n)}))_{n \geq 0}$ to get as close as desired to $\mu(x)$. In fact, there exists some maps that induce the Scott topology everywhere but fail to be continuous, for example, taking $P \coloneqq ([0,1],\leq_{op})$, $\mu_0:P \to [0,\infty)^{op}$, where $0 \mapsto 0$ and $x \mapsto x + \frac{1}{2}$. We continue defining,  analogously, complexity for functions between effectively given domains $f:D \rightarrow D'$.

\begin{defi}[Function complexity]
\label{def:function complex}
Take $f:D \rightarrow D'$ a computable function between dcpos $D,D'$ with effectively given bases $(b_n)_{n\geq0}$ and $(b'_n)_{n\geq0}$, respectively, and $\mu$ a map inducing the Scott topology everywhere in $D'$. We say the time (space) complexity of $f$ is bounded by $t:\mathbb{N} \rightarrow \mathbb{N}$ if there exists a computable function $\phi:\mathbb{N} \to \mathbb{N}$ 
which computes $\phi(n)$ in $t(n)$ steps (using $t(n)$ cells) for all $n \in \mathbb{N}$
and we have both $\mu(b'_{\phi(n)})-\mu(f(b_{\pi_1(n)})) < 2^{-\pi_2(n)}$ for all $n \in \mathbb{N}$ and $\sqcup (b'_{\phi(n)})_{n \in \{\langle m,p \rangle| p \in \mathbb{N}\}} = f(b_m)$ for all $m \in \mathbb{
N}$.
\end{defi}

Notice, we can use the same dcpo $P$ we specified after Definition \ref{def:element complex} to introduce complexity for functions between the real numbers according to Definition \ref{def:function complex}. However, for the sake of brevity, we do not include any specific example of function complexity here. 

Definition \ref{def:function complex} follows the same reason behind Definition \ref{def:element complex}, although with a global approach to approximating the image through $f$ of all the elements in some basis $B \subseteq D$. 
Note that we also use Proposition \ref{compu sets} in Definition \ref{def:function complex}. We can do so since we restrict the definition of complexity to computable functions, which preserve the intuition we aim to capture (see Proposition \ref{why << needed}).

Any computational resource of interested other than time and space can be used in Definitions \ref{def:element complex} and \ref{def:function complex}.
In the particular case of time, the previous definitions allow us to address polynomial computability. 

\begin{defi}[Polynomial-time computable elements and functions]
Either a computable element $x \in P$ or a computable function $f:D \rightarrow D'$  is said to be polynomial-time computable if its time complexity is bounded by a polynomial $p:\mathbb{N} \rightarrow \mathbb{N}$.
\end{defi}

Now that we have introduced the complexity in the uniform approach, we can address it in the non-uniform framework. Since the characterization of elements by their properties is given by the Scott topology in both the uniform and the non-uniform approach, we require the existence of continuous maps inducing the Scott topology everywhere in both of them in order to define the complexity of elements. It is here, however, where we encounter another difference between these approaches. As we show in Proposition \ref{existence map inducing} in the Appendix \ref{inducing}, every dcpo with a countable basis has a continuous map that induces the Scott topology everywhere \cite[Theorem 2.5.1]{martin2000foundation}. There, it is key that the Scott topology is second countable for any dcpo with a countable basis. However, this is no longer true for dcpos where only countable weak bases exist, as we show in Proposition \ref{prop no weak complex}. Moreover, in this scenario, continuous maps that induce the Scott topology everywhere do not exist. 

\begin{prop}
\label{prop no weak complex}
The following statements hold:
\begin{enumerate}[label=(\roman*)]
\item There exist dcpos where effective weak bases exist and the Scott topology is not second countable.
\item Moreover, there exist dcpos where effective weak bases exist and continuous maps that induce the Scott topology everywhere do not.
\end{enumerate}
\end{prop}

\begin{proof}
$(i)$ Take $P \coloneqq ( (I_n)_{n\geq 0} \cup \{p\}, \preceq)$, where $I_n \coloneqq \mathbb{N}$ for all $n \geq 0$ and we have for all $x,y \in P$
\begin{equation*}
    x \preceq y \iff 
    \begin{cases}
    x \leq y &\text{ if } x,y \in I_n \text{ for some } n\geq 0,\\
    y=p.
    \end{cases}
\end{equation*}
(See Figure \ref{no weak complexity} for a representation of $P$.) Note $B \coloneqq (I_n)_{n\geq 0}$ is an effective weak basis. We conclude showing there are no countable bases of $\sigma(P)$. In order to do so, we first characterize the open sets in $P$. Take, thus, some $O \in \sigma(P)$, $O \neq \emptyset$. Since $O$ is upper closed, we have $p \in O$. Moreover, since any directed set $D \subseteq P$, for which $\sqcup D \in O$ holds, fulfills $D \cap O \neq \emptyset$, there exists some $b_n \in I_n$ such that $b_n \in O$ for all $n \geq 0$. Consider, for all $n \geq 0,$ $b^O_n$ the minimal element in $I_n$ that belongs to $O$ and note $O = \cup_{n \geq 0} [b^O_n,p]$, where $[x,p] \coloneqq \{y \in P| x \preceq y \preceq p\}$ for all $x \in P$. Assume, hence, $(O_m)_{m\geq0}$ is a countable basis of $\sigma(P)$, where, as we just showed, $O_m = \cup_{n \geq 0} [b^m_n,p]$ with $b^m_n \in I_n$ for all $n,m \geq 0$. Define, then, $O' \coloneqq \cup_{n \geq 0} [b'_n,p]$, where $b'_n \in I_n$ and $b^n_n < b'_n$ for all $n \geq 0$. Note $O' \in \sigma(P)$ and $O' \neq \emptyset$. Take, thus, some $y \in O'$ and note $O_m \not \subseteq O'$ for all $m \geq 0$ by construction, since $b^m_m \not \in O'$. This contradicts the fact $(O_m)_{m \geq 0}$ is a countable basis of $\sigma(P)$.

$(ii)$ Take $P$ the dcpo in $(i)$ and assume there exists a continuous map which induces the Scott topology everywhere $\mu$. We show this leads to contradiction. Fix, thus, such a $\mu$ and assume w.l.o.g. $\mu(x)=0$. Consider $(q_n)_{n\geq0}$ a numeration of the rational numbers in the interval $(0,\infty)$. Since $\sqcup I_n=p$ for all $n \geq 0$ and $p \in \mu^{-1}([0,q_n)) \in \sigma(P)$ given that $\mu$ is continuous, there exist some $m_n \in I_n$ such that $\mu(m_n) < q_n$ for all $n \geq 0$. Pick, in particular, $m_n$ to be the minimal element in $I_n$ with this property and consider the open set $O = \cup_{n \geq 0} [m_n +1,p]$. Since $p \in O$ and $\mu$ induces the Scott topology everywhere, there exists some $\varepsilon>0$ such that $p \in \mu_\varepsilon (p) \subseteq O$. In particular, we have $\mu_{q_n} (p) \subseteq \mu_\varepsilon (p) \subseteq O$ for all $n \geq 0$ such that $\mu(x)<q_n<\varepsilon$. However, fixing such a $q_n$, we have $m_n \in \mu_{q_n} (p)/O$. Since we can argue analogously for any $\varepsilon >0$, we have a contradiction and, hence, there is no continuous map that induces the Scott topology everywhere in $P$.
\end{proof}

\begin{figure}[!tb]
\centering
\begin{tikzpicture}
 \node[rounded corners, draw,fill=blue!10, text height = 10cm,yshift=4.2cm, minimum width = 2cm,label={[anchor=south,below=.1cm]270:\textbf{$I_n$}}] (8) {};
    \node[other node] (1) {$0_n$};
    \node[rounded corners, draw,fill=blue!10, text height = 10cm, minimum width = 2cm,right = .75cm of 1, yshift=4.2cm,label={[anchor=south,below=.1cm]270:\textbf{$I_{n+1}$}}] (9) {};
    \node[other node] (2) [ right = 1cm of 1]  {$0_{n+1}$};
     \node[rounded corners, draw,fill=blue!10, text height = 10cm, minimum width = 2cm, yshift=4.2cm,left = 0.75cm of 1,label={[anchor=south,below=.1cm]270:\textbf{$I_{n-1}$}}] (10) {};
    \node[other node] (3) [ left = 1cm  of 1]  {$0_{n-1}$};
    \node[other node] (4) [above  = 1cm  of 1]  {$1_n$};
    \node[other node] (5) [above  = 1cm  of 2]  {$1_{n+1}$};
    \node[other node] (6) [above  = 1cm  of 3]  {$1_{n-1}$};
    
    \node[other node] (11) [above  = 1cm  of 4]  {$n_{n}$};
    \node[other node] (12) [above  = 1cm  of 5]  {$n_{n+1}$};
    \node[other node] (13) [above  = 1cm  of 6]  {$n_{n-1}$};
     \node[other node] (7) [above  = 4cm  of 11]  {$p$};
     
       \node[other node] (14) [above  = 1cm  of 11]  {${n+1}_{n}$};
         \node[other node] (15) [above  = 1cm  of 12]  {${n+1}_{n+1}$};
           \node[other node] (16) [above  = 1cm  of 13]  {${n+1}_{n-1}$};
     
     \node[smaller dot] [above  = 0.25cm  of 4]  {};
     \node[smaller dot] [above  = 0.5cm  of 4]  {};
     \node[smaller dot] [above  = 0.75cm  of 4]   {};
     
      \node[smaller dot] [above  = 0.25cm  of 5]  {};
     \node[smaller dot] [above  = 0.5cm  of 5]  {};
     \node[smaller dot] [above  = 0.75cm  of 5]   {};
     
      \node[smaller dot] [above  = 0.25cm  of 6]  {};
     \node[smaller dot] [above  = 0.5cm  of 6]  {};
     \node[smaller dot] [above  = 0.75cm  of 6]   {};
     
      \node[smaller dot] [above  = 0.25cm  of 14]  {};
     \node[smaller dot] [above  = 0.5cm  of 14]  {};
     \node[smaller dot] [above  = 0.75cm  of 14] (C)  {};
     
      \node[smaller dot] [above  = 0.25cm  of 15]  {};
     \node[smaller dot] [above  = 0.5cm  of 15]  {};
     \node[smaller dot] [above  = 0.75cm  of 15] (A) {};
     
      \node[smaller dot] [above  = 0.25cm  of 16]  {};
     \node[smaller dot] [above  = 0.5cm  of 16]  {};
     \node[smaller dot] [above  = 0.75cm  of 16] (B) {};
     
     \node[smaller dot] [left  = 0.5cm  of 13]  {};
     \node[smaller dot] [left  = 0.75cm  of 13]  {};
     \node[smaller dot] [left  = 1cm  of 13]  {};
     
     \node[smaller dot] [right  = 0.5cm  of 12]  {};
     \node[smaller dot] [right = 0.75cm  of 12]  {};
     \node[smaller dot] [right = 1cm  of 12]  {};

   \path[draw,thick,->]
    (1) edge node {} (4)
    (2) edge node {} (5)
    (3) edge node {} (6)
    (11) edge node {} (14)
    (12) edge node {} (15)
    (13) edge node {} (16)
    ;
    
    \path[draw,thick,dotted]
    (A) edge node {} (7)
    (B) edge node {} (7)
    (C) edge node {} (7)
    ;
\end{tikzpicture}
\caption{Dcpo introduced in Proposition \ref{prop no weak complex} where, although countable weak bases exist, the Scott topology is not second countable (in fact, not even first countable) and, moreover, there is no continuous map that induces the Scott topology everywhere. In particular, we represent three elements of $I_{n-1}$, $I_n$ and $I_{n+1}$ for some $n>2$ and their relation through arrows. The thick dotted lines represent eiyher the missing elements in $I_{n-1}$, $I_n$ and $I_{n+1}$ or $I_m$ for both $m<n-1$ and $n+1<m$, while the thin dotted ones represent the fact all elements in $I_{n-1}$, $I_n$ and $I_{n+1}$ are below $p$. The arrows unite elements with those that are immediately above them.}
\label{no weak complexity}
\end{figure}
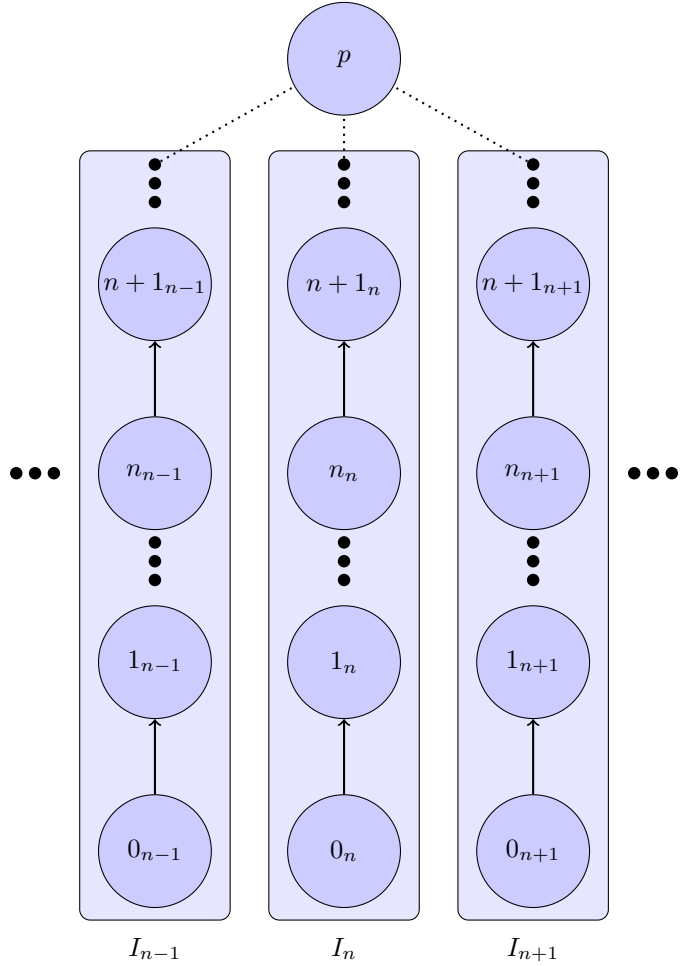

Note $(ii)$ in Proposition \ref{prop no weak complex} implies $(i)$ by Proposition \ref{existence map inducing} $(i)$. Note, also, $P$ in Proposition \ref{prop no weak complex} is an example of a dcpo with an effective weak basis whose Scott topology is not first countable, since any local basis for $p$ is a basis for $\sigma(P)$.

In summary, while there exists a complexity framework for both elements and functions in the uniform approach, there does not seem to be one in the non-uniform framework, since, apparently, a proper notion of precision does not always exist there (as we showed in Proposition \ref{prop no weak complex}).

\section{Conclusion}

We have introduced a more general order-theoretic framework which allows to define computability on uncountable spaces and have compared it with the previous approach in domain theory. We have pointed out differences $(i)-(v)$ in Section \ref{comparing apporaches}
and explored them using four criteria: computable elements (Section \ref{diff comp ele}), computable functions (Section \ref{diff comp func}), model dependence of computability (Section \ref{model dep comp}) and complexity (Section \ref{complexity}).
Our conclusions regarding these differences can be summarized as follows:

\begin{enumerate}[label=(\roman*)]
\item Although the way-below relation $\ll$ reduces the number of ordered spaces where computable elements can be introduced, 
it is essential in order to $(a)$ have a definition of
computable functions $f$ that preserves the intuition behind them, namely, that computable elements should be mapped to computable elements and, more specifically, that one should be able to construct an effective procedure for $f(x)$ when given one for $x$, $(b)$ to have a notion of computability of elements and functions that is likely to be independent of the chosen computability model, that is, basis and finite map, and $(c)$ to have a notion of complexity for both elements and functions.
\item Computable elements can be defined in the more general framework, which covers a wider family of ordered sets. Even though their definition seems to vary between the non-uniform an the uniform approach, they coincide under the assumptions in the uniform one. 
\item While a definition of computable functions that captures the intuition behind them can be defined in the uniform approach, we argue it seems unlike it can also be done in the more general framework, that is, that some extra structure like $\ll$ is needed.
\item The inclusion of a bottom element $\perp$ is needed in order for certain functions that undergo comparisons to be computable. 
That is, some functions will perform certain comparisons on some inputs which, if favorable, determine what the output should be. Nonetheless, whenever such comparisons fails, it is useful to know in advance certain output will always be correct for some family of computations. This way, when the comparisons fail, the function can output that correct element and terminate. This is the role $\perp$ and also the reason why $\perp=b_0$ is fixed.
The inclusion of $\perp$ is relevant to show both definitions of computable elements coincide in the uniform picture, to have a proper definition of computable functions and to reduce the dependence of computability on the model. 
\item The strong version of effectivity in the uniform approach, that is, the fact that $\{ \langle n,m \rangle | b_n \ll b_m\}$ is recursively enumerable, is needed in order to assure that the definition of computable element in both approaches coincide in the uniform framework. This contrasts with the natural extension to the uniform case of the definition of effectivity in the non-uniform approach, which would fail to have this property.
\end{enumerate}

Aside from these differences, we have emphasized the role of finite map, which was not made explicit previously. It was, however, implicit in domain theory, since the elements in a basis are regarded as \emph{finite}, in the sense of having a finite representation \cite{martin2000foundation}.
Finally, we introduced complexity notions in domain theory which, to our best knowledge, were missing. In order to do so, we have translated the main ideas in \cite{ko1998polynomial,ko2012complexity}, where real numbers are addressed, to the framework of domain theory.

\begin{appendix}
\section{Appendix}
\subsection{Proofs}
\label{proofs}
\begin{prop}
\label{cardinality where computability}
If $P$ is a dcpo with a countable weak basis and $\mathfrak{c}$ is the cardinality of the continuum, then $|P| \leq \mathfrak{c}$.
\end{prop}

\begin{proof}
Take $B=(b_n)_{n\geq0}$ a countable weak basis for $P$ and consider the map $\phi: P \to \mathcal{P}(\mathbb{N})$, $x \mapsto \{n \in \mathbb{N}| b_n \preceq x\}$. If we show $\phi$ is injective, then we have finished, since we get $|P| \leq |\mathcal{P}(\mathbb{N})|= \mathfrak{c}$.
Consider, thus, $x,y \in P/B$ such that $\phi(x)=\phi(y)$. By definition of weak basis, there exist a couple of directed set $B_x \subseteq \phi(x)$, $B_y \subseteq \phi(y)$ such that $\sqcup B_x=x$ and $\sqcup B_y=y$. Since $\phi(x)=\phi(y)$, we have $b \preceq y$ $\forall b \in B_x$, thus, by definition of supremum, $x \preceq y$. Analogously, we have $y \preceq x$. By antisymmetry, we conclude $x=y$, hence, $\phi$ is injective.  
\end{proof}

\begin{prop}
\label{example}
$\Sigma^*$ is an effective weak basis for the Cantor domain $(\Sigma^{\infty},\preceq_C)$.
\end{prop}

\begin{proof}
It is easy to see $\Sigma^*=(b_n)_{n\geq0}$ is a countable weak basis for the Cantor domain, given that $\sqcup \{y \in \Sigma^*| y \preceq_C x\}=x$ for all $x \in \Sigma^\omega$. To conclude, we show $\{\langle n,m \rangle| b_n \preceq_C b_m\}$ is recursively enumerable. We need to construct, thus, a computable function $f:\mathbb{N} \to \mathbb{N}$ such that $f(\mathbb{N})=\{\langle n,m \rangle| b_n \preceq_C b_m\}$. In the following, we assume $\Sigma=\{0,1\}$, although we could consider any other alphabet analogously. We begin constructing a finite map $\alpha: \mathbb{N} \to \Sigma^*$ which orders the strings in $\Sigma^*$ by length and, whenever they have the same length, interprets them as binary representations of the natural numbers and orders them accordingly. That is,
if $n=0,1$, then $\alpha(n)=n$. Otherwise, we begin with $m=2$ and continue in increasing order until we find some $m \geq 2$ such that $n< \sum_{i=1}^m 2^i$. We take then $n'= n-(1+ \sum_{i=1}^{m-1} 2^i)$ and the binary string representation of $n'$ as $\alpha(n)$, adding zeros to the left if needed until we have $m$ digits.
We construct now $f$ specifying its associated Turing machine $M$. Given input $n \in \mathbb{N}$, $M$ begins computing $m=\pi_1(n)$ and $p=\pi_2(n)$, since $\pi_1$ and $\pi_2$ are computable,
and then turns $m$ and $p$ into the finite strings in associated to them $\alpha(m)$ and $\alpha(p)$ following the recipe we exposed before. Once $\alpha(m)$ and $\alpha(p)$ are known, $M$ compares  their digits sequentially from right to left
until $\alpha(m)$ has no digits left. If certain digit does not coincide or $\alpha(p)$ has no more digits, $M$ outputs $0$ and terminates. Otherwise, it outputs $n$ and terminates. Note we use $0$ as output for the incorrect cases as $0=\langle 0,0\rangle$ belongs to $f(\mathbb{N})$.
\end{proof}

\begin{prop}
\label{count compu f}
If $P$ and $Q$ are dcpos with effective bases $B$ and $B'$, respectively,
then the set of $(B,B')$-computable functions $f:P \to Q$ is countable.
\end{prop}

\begin{proof}
Let $B=(b_n)_{n\geq0}$ and $B'=(b'_n)_{n\geq0}$ be, respectively, enumerations of these bases of $P$ and $Q$. Define the sets $C \coloneqq \{f:P \to Q| f \text{ computable}\}$ and $C_{\mathbb{N}} \coloneqq \{g: \mathbb{N} \to \mathbb{N}| g \text{ computable}\}$ and
consider the map 
 \begin{align*}
\phi: C&\rightarrow C_\mathbb{N}\\
f&\mapsto f_g,
\end{align*}
 where $f_g(\mathbb{N}) \coloneqq \{ \langle n,m \rangle| b'_n \ll f(b_m)\}$ $\forall f \in C$. We will show $\phi$ is injective, which implies
 $|C| \leq |C_\mathbb{N}| \leq |\mathbb{N}|$, where the last inequality is a well-known fact about Turing machines \cite{rogers1987theory}.
 Take, thus, $f,f' \in C$ and
 assume we have $f_g=f'_g$. By definition, we have $b'_n \ll f(b_m)$ if and only if $b'_n \ll f'(b_m)$ $\forall m,n \in \mathbb{N}$ and, as a result,
 \begin{equation*}
     f(b_m)= \sqcup (\twoheaddownarrow f(b_m) \cap B') = \sqcup (\twoheaddownarrow f'(b_m) \cap B')= f'(b_m)
 \end{equation*}
 $\forall m \in \mathbb{N}$, where we applied \cite[Proposition 2.2.4]{abramsky1994domain}. Take now some $x \in P/B$. Since we have $x = \sqcup \{b_n| b_n \ll x\}$ and $f$ is continuous in $\sigma(P)$ by definition, we conclude
 \begin{equation*}
     f(x) = \sqcup \{f(b_n) | b_n \ll x\} = \sqcup \{f'(b_n)|b_n \ll x\} = f'(x)  
 \end{equation*}
 $\forall x \in P/B$. As a result, we have $f=f'$ whenever $f_g=f'_g$ and, thus, $\phi$ is injective.
\end{proof}

\subsection{Maps inducing the Scott topology everywhere}
\label{inducing}

In order to follow the definition of complexity on the real numbers in \cite{ko1998polynomial}, we used a numerical representation of precision in dcpos
in both Definitions \ref{def:element complex} and \ref{def:function complex}. We took as such maps inducing the Scott topology everywhere. We explore them briefly in this section.

We begin, in Proposition \ref{existence map inducing}, showing such maps exist wherever the uniform order-theoretic approach can be applied,
that is, for $w$-continuous dcpos.

\begin{prop}
\label{existence map inducing}
The following statements hold:
\begin{enumerate}[label=(\roman*)]
\item \cite{martin2000foundation} If $P$ is a dcpo whose Scott topology is second countable, then $P$ has a continuous map that induces the Scott topology everywhere. In partiuclar, the same conclusion holds for $w$-continuous dcpos.
\item There are continuous dcpos $P$ where continuous maps that induce the Scott topology everywhere exist and countable bases for $\sigma(P)$ do not.
\end{enumerate}
\end{prop}

\begin{proof}
$(i)$ We show it for $\omega$-continuous dcpos, although the more general case of dcpos with a second countable Scott topology can be shown analogously. Given a countable basis $(b_n)_{n\geq0}$, define for all $n \geq 0$ $\chi_n(x)\coloneqq 1$ if $x \in \twoheaduparrow b_n$ and $\chi_n(x)\coloneqq 0$ otherwise. We will show $\mu(x) \coloneqq 1-\sum_{n\geq 0} 2^{-(n+1)} \chi_n(x)$ is a map inducing the Scott topology everywhere. We note first it is a monotone, as $b_n \ll x \preceq y$ implies $b_n \ll y$ for all $ n \geq 0$. If $x \in O \in \sigma(P)$, then there exists some $m \geq 0$ such that $x \in \twoheaduparrow b_m \subseteq O$, since $(\twoheaduparrow b_n)_{n\geq0}$ is a basis for $\sigma(P)$ \cite{abramsky1994domain}. Take, then, $\varepsilon \coloneqq 1-\sum_{n=0}^{m}2^{-(n+1)} \chi_n(x)$ and some $y \in P$ such that $y \preceq x$ and $\mu(y)<\varepsilon$. We want to prove $y \in O$. By definition, we have
\begin{equation*}
    \sum_{n\geq 0} 2^{-(n+1)} \chi_n(y) > \sum_{n=0}^{m}2^{-(n+1)} \chi_n(x).
    \end{equation*}
Since $y \preceq x$, we have, for all $n \geq 0$, $b_n \ll x$ whenever $b_n \ll y$. There is, thus, some $C \geq 0$ such that
    \begin{equation*}
    \sum_{n\geq m} 2^{-(n+1)} \chi_n(y)  > 2^{-(m+1)} + C.
\end{equation*}
To conclude, assume $y \not \in \twoheaduparrow b_m$. Then, we would have $\sum_{n\geq m+1} 2^{-(n+1)} \chi_n(y)  > 2^{-(m+1)}$, which is a contradiction since $\sum_{n\geq m+1} 2^{-(n+1)} \chi_n(y)  \leq  \sum_{n\geq m+1} 2^{-(n+1)}$ $=2^{-(m+1)}$. Thus, $y \in \twoheaduparrow b_m \subseteq O$. We obtain for each $x \in O \in \sigma(P)$ there exists some $\varepsilon >0$ such that $\mu_{\varepsilon}(x) \subseteq O$, that is, $\mu$ induces the Scott topology everywhere. To conclude, we ought to see $\mu$ is continuous. Since it is a monotone, we only need to show, given a directed set $D$, we have $\sqcup \mu(D)=\mu(\sqcup D)$. Assume, thus, there exists some $z \in [0,\infty)$ such that $\mu(d) \leq_{op} z <_{op} \mu(\sqcup D)$ for all $d \in D$. Then, there exists some $N\geq 0$ such that $z <_{op}1-\sum_{n=0}^N 2^{-(n+1)} \chi_n(x)$. Since there exits some $d_n \in D$ such that $d_n \in \twoheaduparrow b_n$ whenever $0 \leq n \leq N$, $D$ is directed and $\twoheaduparrow b_n$ is upper closed,
we can
follow \cite{hack2022relation} and
find some $d' \in D$ such that $d' \in \twoheaduparrow b_n$ for all $n$ such that $0 \leq n \leq N$. Hence, $z <_{op} \mu(d')$, contradicting the definition of $z$. As a results, if there exists some $z \in [0,\infty)$ such that $\mu(d) \leq_{op} z$ for all $d \in D$, then $z \leq_{op} \mu(\sqcup D)$ holds. Thus, $\mu$ is continuous.

$(ii)$ For the second statement, we simply take $P \coloneqq (\mathbb{R},\preceq)$, where $\preceq$ is the trivial partial order. Note $P$ is continuous, since $K(P)=P$ and $\sigma(P)$ is not second countable since $\{x\} \in \sigma(P)$ for all $x \in P$ and, thus, any basis of $\sigma(P)$ is uncountable. To conclude, consider the map $\mu: P \to [0,\infty)^{op}$, $x \mapsto e^x$ and note $\mu$ induces the Scott topology everywhere, since any map of $P$ is a monotone and, given some $x \in O \in \sigma(P)$, we can take any $\varepsilon$ such that $\varepsilon > \mu(x)$ and get $\mu_\varepsilon(x)=\{x\} \subseteq O$. Note, also, one can directly see $\mu$ is continuous. 
\end{proof}

Note the first statement in Proposition \ref{existence map inducing} is already in \cite{martin2000foundation}. We have, however, showed it in a direct fashion while, in order to do so there, either some extra machinery was needed \cite[Theorem 2.5.1]{martin2000foundation} or  composition of maps inducing the Scott topology for two different spaces was used \cite[Example 2.5.1 and Example 2.5.3] {martin2000foundation}.

A map inducing the Scott topology is simply a method which allows us to measure how many of the properties that define some $x$ are known at certain stage of a computation leading to $x$. We can see this in the proof of Proposition \ref{existence map inducing}, where $\mu(y)$ is simply a weighted sum of the presence of $y$ in a  countable family of sets that differentiate the elements in the dcpo. 
While we know, then, any space where we introduce uniform computation has functions which allow us to numerically measure precision, the one constructed in Proposition \ref{existence map inducing} is impractical as it relies on an infinite series. Given a specific $w$-continuous dcpo, finding out whether a convenient, that is, simple to evaluate, map induces the Scott topology everywhere is regarded as a difficult task \cite{martin2008technique}. To conclude, we show, in Proposition \ref{upper comp and measurement}, such a difficulty is absent in conditionally connected dcpos, where maps inducing the Scott topology everywhere coincide with strict monotones. Note a partial order $P$ is conditionally connected if, for all $x,y,z \in P$ such that $x,y \preceq z$, we either have $x \preceq y$ or $y \preceq x$
\cite{hack2022relation}
and a monotone $v:P \rightarrow \mathbb{R}$ is a \emph{strict monotone} if $x \prec y$ implies $v(x) < v(y)$ \cite{alcantud2016richter}. Moreover, we say $v$ is \emph{lower semicontinuous} provided $v((t,\infty))^{-1} \in \sigma(P)$ for all $t \in \mathbb{R}$. Note the Cantor domain is conditionally connected. 

\begin{prop}
\label{upper comp and measurement}
If $P$ is a conditionally connected dcpo, then there exists a (continuous) map $\mu:P \rightarrow [0,\infty)^{op}$ which induces the Scott topology everywhere in $P$ if and only if there exists a (lower semicontinuous) strict monotone $v:P \rightarrow \mathbb{R}$.
\end{prop}
\begin{proof}
We begin showing any monotone map $\mu:P \rightarrow [0,\infty)^{op}$ induces the Scott topology everywhere in $P$ if and only if $x \prec y$ implies $\mu(x) <_{op} \mu(y)$. Note the fact 
that maps inducing the Scott topology everywhere $\mu:P \rightarrow [0,\infty)^{op}$ fulfill that $x \prec y$ implies $\mu(x) <_{op} \mu(y)$ is already known \cite[Lemma 2.2.5]{martin2000foundation}. For the converse, take some $x \in P$ and $O \in \sigma(P)$ such that $x \in O$. Since $P$ has a basis,
by \cite[Proposition 14]{hack2022relation},
there exists some $y \in O$ such that $y \ll x$. If $x \in K(P)$, then $x \in I(P) \cup min(P)$
by \cite[Proposition 12]{hack2022relation}.
In case $x \in min(P)$, then we take some $\varepsilon > \mu(x) $ and we have finished, since for all $y \in P$ we have $y \preceq x$ if and only if $x=y$. If $x \in I(P)$, then there exists some $v_x \in P$, $v_x \prec x$, such that, whenever $y \prec x$ holds, then $y \preceq v_x$. Since $\mu$ is strictly monotonic, take some $\varepsilon >0$ such that $\mu(x) < \varepsilon < \mu(v_x)$. By construction, $\mu_{\varepsilon}(x) = \{x\} \subseteq O$, since any $y \prec x$ fulfills $y \preceq v_x$ and, by monotonicity, $\varepsilon<\mu(v_x) \leq \mu(y)$. If $x \notin K(P)$, then there exists some $y \in O$ such that $y \prec x$, since $P$ has a basis. Take, thus, some $\varepsilon >0$ such that $\mu(x) < \varepsilon < \mu(y)$ and some $a \in P$ such that $a \preceq x$ and $\mu(a) < \varepsilon$. By upper comparability, we have $\neg(a \bowtie y)$. If $a \preceq y$, then $\mu(a) \geq \mu(y)$ by monotonicity which leads to contradiction as $\mu(a) < \varepsilon < \mu(y)$. Thus, $y \prec a$ and $a \in O$, since $O$ is upper closed by definition. In summary, $\mu$ induces the Scott topology everywhere in $P$.

To conclude, we need to take care of continuity and translate strict monotonicity from $\mu$ to some $v:P \to \mathbb{R}$. First, note on $[0,\infty)^{op}$ the Scott topology consists on sets of the form $[0,a)$ with $a>0$. Given some $v:P \rightarrow \mathbb{R}$, we can take some strictly decreasing function $f:\mathbb{R} \rightarrow [0, \infty)$ and consider $\mu:=f \circ v$, where $\mu^{-1}\big([0,a)\big)=v^{-1}\big((f^{-1}(a),\infty)\big) \in \sigma(P)$, implying $\mu$ is continuous. Since $f$ is strictly decreasing, the monotonicity properties of $v$ for $\leq$ are carried to $\mu$ for $\leq_{op}$. Analogously, given $\mu:P \rightarrow [0,\infty)^{op}$, we can find some strictly decreasing function $g:[0,\infty) \rightarrow \mathbb{R}$  and take $v:= g \circ \mu$, where $v^{-1}\big((t,\infty)\big) = \mu^{-1}\big([0, g^{-1}(t)\big) \in \sigma(P)$, implying $v$ is lower semicontinuous. Since $g$ is strictly decreasing, the monotonicity properties of $\mu$ for $\leq_{op}$ are carried to $v$ for $\leq$. 
\end{proof}

As an example where this can be applied, we return to the Cantor domain. We can consider the following map for $\Sigma^\infty$ \cite{martin2000foundation}
\begin{align*}
\ell:\Sigma^\infty&\rightarrow [0, \infty)^{op}\\
s&\mapsto 2^{-|s|}
\end{align*}
where $|\cdot|: \Sigma^\infty \rightarrow \mathbb{N} \cup \{ \infty\}$ assigns to every string its length. Since $\ell$ is clearly strictly monotonic and $\Sigma^\infty$ is conditionally connected, $\ell$ induces the Scott topology everywhere. Note we have avoided showing for each  $s \in \Sigma^\infty$ and $O \in \sigma(\Sigma^\infty)$ such that $s \in O$ there exists some $\varepsilon >0$ such that $\ell_\varepsilon(s) \subseteq O$. While in this particular instance it is easy to do it, in other instances it is not \cite{martin2008technique}.

\end{appendix}

\newpage
\bibliographystyle{plain}
\bibliography{main}

\end{document}